\documentclass[twocolumn,twoside]{IEEEtran}
\usepackage{times} % assumes new font selection scheme installed
\usepackage{epsfig}
\usepackage{graphicx}
\usepackage{epstopdf}
\usepackage{algorithm,algorithmic,color,subfigure, latexsym, amsmath, amsfonts, amssymb, cite}
\usepackage{algorithm}
\usepackage{algorithmic}
\usepackage{epsfig}
\usepackage{graphicx}
\usepackage{epstopdf}
\usepackage{bm,amsmath,amssymb,amsfonts,graphicx,epsfig,amsthm,color}
\usepackage{bbold,dsfont}
\usepackage{float}

\usepackage[hyphens]{url}
\PassOptionsToPackage{bookmarks=false}{hyperref}

\usepackage[depth=-1]{bookmark}

\usepackage{booktabs}       % professional-quality tables
\usepackage{amsfonts}       % blackboard math symbols
\usepackage{microtype}      % microtypography

\usepackage{times}
\usepackage{graphicx} % more modern
\usepackage{wrapfig}



\usepackage{accents}
\makeatletter
\def\wid{\check{{\cc@style\underline{\mskip9.5mu}}}}
\def\Wideubar{\underaccent{{\cc@style\underline{\mskip6mu}}}}
\makeatother

\makeatletter
\def\wideubar{\underaccent{{\cc@style\underline{\mskip9.5mu}}}}
\def\Wideubar{\underaccent{{\cc@style\underline{\mskip6mu}}}}
\makeatother

\makeatletter
\def\widebar{\accentset{{\cc@style\underline{\mskip9.5mu}}}}
\def\Widebar{\accentset{{\cc@style\underline{\mskip6mu}}}}
\makeatother

\newtheorem{lemma}{Lemma}

\newtheorem{theorem}{Theorem}

\newtheorem{assumption}{Assumption}
\theoremstyle{remark}\newtheorem{remark}{Remark}

\allowdisplaybreaks
\title{\LARGE \bf
	Resilient Control under Quantization and Denial-of-Service: Co-designing a Deadbeat Controller and Transmission Protocol 
}

\author{Wenjie~Liu,~Jian~Sun,~\IEEEmembership{Senior Member,~IEEE}, Gang Wang,~\IEEEmembership{Member,~IEEE},\\Francesco Bullo,~\IEEEmembership{Fellow,~IEEE},~and
	Jie Chen,~\IEEEmembership{Fellow,~IEEE}% <-this % stops a space
	\thanks{This work was supported in part by the National Natural Science Foundation of China under Grants 61925303, 62088101, U20B2073, 61720106011, and the National Key R$\&$D Program of China under Grant 2018YFB1700100. 
	}
	\thanks{W. Liu and G. Wang are with the State Key Lab of Intelligent Control and Decision of Complex Systems and the School of Automation, Beijing Institute of Technology, Beijing 100081, China (e-mail: liuwenjie@bit.edu.cn; gangwang@bit.edu.cn).
		
		J. Sun is with the State Key Lab of Intelligent Control and Decision of Complex Systems and the School of Automation, Beijing Institute of Technology, Beijing 100081, China, and the Beijing Institute of Technology Chongqing Innovation Center, Chongqing 401120, China (e-mail: sunjian@bit.edu.cn).
		
		J. Chen is with the Department of Control Science and Engineering, Tongji University, Shanghai 201804, China, and also with the State Key Lab of Intelligent Control and Decision of Complex Systems and the School of Automation, Beijing Institute of Technology, Beijing 100081, China 	
		(e-mail: chenjie@bit.edu.cn).
		
		F. Bullo is with the Mechanical Engineering Department and the Center of Control, Dynamical Systems and Computation, UC Santa Barbara, CA 93106-5070, USA (e-mail: bullo@ucsb.edu).
	}
}

\begin{document}
	\maketitle

	%%%%%%%%%%%%%%%%%%%%%%%%%%%%%%%%%%%%%%%%%%%%%%%%%%%%%%%%%%%%%%%%%%%%%%%%%%%%%%%%
	\begin{abstract}
		This paper is concerned with the problem of stabilizing continuous-time linear time-invariant systems subject to quantization and
		Denial-of-Service (DoS) attacks.
		In this context, two DoS-induced challenges emerge with the design of resilient encoding schemes, namely, the coupling between encoding strategies of different signals, and the synchronization between the encoder and decoder.
		To address these challenges, a novel structure that is equipped with a deadbeat controller as well as a delicate transmission protocol for the input and output channels, co-designed leveraging the controllability index, is put forward.
		When both input and output channels are subject to DoS attacks and quantization, the proposed structure is shown able to decouple the encoding schemes for input, output, and estimated output signals.
		This property is further corroborated by designing encoding schemes as well as conditions that ensure exponential stability of the closed-loop system.
		On the other hand, when only the output channel is subject to network phenomenon, the proposed structure can achieve exponential stabilization without acknowledgment (ACK) signals, in contrast to existing ACK-based results.  
		Finally, a numerical example is given to demonstrate the practical merits of the proposed approach as well as the theory.
	\end{abstract}
	\begin{keywords} Denial-of-Service attacks, quantization, deadbeat control, acknowledgment-free protocol.
	\end{keywords}
	
	%%%%%%%%%%%%%%%%%%%%%%%%%%%%%%%%%%%%%%%%%%%%%%%%%%%%%%%%%%%%%%%%%%%%%%%%%%%%%%%%
	\section{Introduction}~\label{sec:intro}
	Driven by recent advances in computing, communication, and networking technologies, modern engineering systems (e.g., \cite{tac2020wwsc,2014Aunmanned,Lv2020Event}) have gradually shifted their computing and control workload to the cloud, and even edge with data transmitted over wired or wireless networks.
	Despite their flexibility, such network-based control systems (a.k.a., networked control systems) are known vulnerable to cyber threats~\cite{2013Attack,Alvaro2009Research}.
	In fact, existing works have shown that malicious attacks can severely disrupt the
	control performance and even render the system unstable \cite{Jang2014survey}.
	Examples of such failures in widely used safety- and security-critical control systems nowadays could put our lives and even national infrastructure at risk \cite{2011Stuxnet}.

	Several types of cyberattacks have been studied, including replay attacks \cite{Zhu2017replay,replay2018}, false-data injection attacks \cite{FP-RC-FB:10p ,wu2019Switching}, and Denial-of-Service (DoS) attacks \cite{shi2015jamming,Cetinkaya2019overview,hu2020dos}.~Relative to the others, DoS attacks can cause jamming in communication channels with little knowledge of system dynamics.~They are easy to launch and have received considerable attention \cite{L2010Protection}.~For instance, the work \cite{PersisInput} developed a general DoS framework, under which closed-loop system stability can be preserved via state-feedback control, provided certain DoS attack frequency and duration conditions are met.~This result has been extended in several directions, e.g., via output-feedback control in \cite{FengResilient}, as well as considering multiple output channels in \cite{LuInput}.
	
	All the aforementioned works assumed that the communication channels have an infinite data-rate.
	Clearly for real-world engineering systems, this condition is difficult to be met.
	Systems with digital communication channels offer a basic paradigm.
	The problem of limited bandwidth have been studied by accounting for the effect of quantization.
	There is a great deal of research indicating that even without attacks, quantization can compromise system performance \cite{bullo1576851}, which is often addressed by designing suitable encoding schemes and providing enough quantization levels.
	To name a few, for stabilization of systems with quantized measurements, \cite{Liberzon2000Quantized} first introduced the so-called ``zooming-in'' and ``zooming-out'' method.
	Following this work, a number of stabilization encoding schemes have been designed for systems with quantized output feedback in \cite{Sharon2008Input, WakaikiObserver}, and switched systems in \cite{WakaikiStability, Liberzon2014Finite, YangFeedback, Wakaiki2017Stabilization}.
	Recently, a few works have considered these two factors (i.e., quantization and DoS attacks) simultaneously; see \cite{chen2018Event,
		feng2020datarate,liu2020datarate, Feng2020multi, 8880482}.
	The trade-off between system resilience against DoS attacks and data-rate was analyzed in \cite{feng2020datarate}.
	The minimum data-rate for stabilizing a centralized system and a multi-agent system were derived in \cite{liu2020datarate} and \cite{Feng2020multi}, respectively.
	Capitalizing on the zooming-in and -out method, the work \cite{8880482} designed a resilient output encoding scheme for systems whose output channel is subject to DoS attacks and limited data-rate.
	
	The goal of this paper is to stabilize systems with both input (controller-to-plant) and output (plant-to-controller) channels subject to DoS attacks and limited bandwidth.
	To this aim, the quantizer encoding schemes should be carefully designed.
	In the absence of DoS attacks, the work \cite{WakaikiObserver} developed encoding schemes for signals transmitted through both input and output channels.
	However, their schemes cannot be applied here, due to the coupling between encoding strategies for different signals in the presence of DoS attacks.	To overcome this challenge, we put forth a delicate structure, including a deadbeat controller and a transmission protocol. 
	Our protocol requires signals transmitted through the input channel at a higher rate than those through the output channel. Precisely, their transmission rate ratio is exactly the controllability index of the system.
	Its efficacy is corroborated by the possibility to decouple design of different encoding schemes as, well as, establishing closed-loop stability conditions.
	We further apply this structure to stabilize systems with only output channel has network imperfections.
	In this scenario, it is proved that the proposed structure can secure the synchronization between encoder and decoder even without acknowledgments (ACKs), which are required by existing works, e.g., \cite{8880482,feng2020datarate}.
	
	In a nutshell, the main contributions of the present work are summarized as follows.
	\begin{itemize}
		\item[\textbf{c1)}]
%		A structure, consisting of a deadbeat controller and a transmission protocol, is presented to address 
%		the coupling issue between different encoding schemes as well as the asynchronization between the encoder and decoder;
		To cope with the coupling and synchronization issues, a structure consisting of a deadbeat controller and a transmission protocol for input and output channels, co-designed in terms of the controllability index, is advocated.
		\item[\textbf{c2)}]
		Under this structure, the input, output, and estimated output encoding schemes can be designed separately to achieve closed-loop stability when both input and output channels are subject to DoS attacks and quantization;
		and,
		\item[\textbf{c3)}]
		When such network phenomena appear only in the output channel, an encoding scheme is designed such that the system can be stabilized through an ACK-free protocol, that is in sharp contrast to existing ACK-based results. 
	\end{itemize}
	\emph{Notation:}  Denote the set of integers (real numbers) by $\mathbb{Z}$ ($\mathbb{R}$). 
	Given $\alpha \in \mathbb{R}$ or $\alpha \in \mathbb{Z}$, let $\mathbb{R}_{>\alpha}$ ($\mathbb{R}_{\ge \alpha}$) or $\mathbb{Z}_{>\alpha}$ ($\mathbb{Z}_{\ge \alpha}$) denote the set of real numbers or integers greater than (greater than or equal to) $\alpha$.
	Let $\mathbb{N}$ denote the set of natural numbers and $\mathbb{N}_0 := \mathbb{N} \cup \{0\}$.
	For a vector $v = [v_1, v_2, \cdots\!, v_n]^T \in \mathbb{R}^n$, denote its maximum norm by $|v| := \max\{|v_1|, \cdots\!, |v_n|\}$ and the corresponding induced norm of a matrix $M \in \mathbb{R}^{m \times n}$ by $\Vert M\Vert := \sup\{|Mv|:v \in \mathbb{R^n}, |v| = 1\}$.
	\section{Preliminaries and Problem Formulation} \label{problem_Formulation}
	\subsection{Problem formulation}\label{systemdefination}
	In this paper, we study the networked control architecture in Fig. \ref{networkfig}, where a plant is to be stabilized by a remote digital controller over a network subject to DoS attacks.~The plant is described by the following dynamics
	\begin{subequations}\label{continuoussystem}
		\begin{align}
		&\dot{x}(t) = Ax(t) + Bu(t) \label{continuoussysteG_1}\\
		&y(t) = Cx(t) \label{continuoussystem_2}
		\end{align}
	\end{subequations}
	where $x(t) \in \mathbb{R}^{n_x}, u(t) \in \mathbb{R}^{n_u}$, and $y(t) \in \mathbb{R}^{n_y}$ are the state, the control input, and the output, respectively.
%	The pair $(A,B)$ is assumed stabilizable, and the pair $(C,A)$ is observable.
	Here, we consider output signals and control inputs to be transmitted through different channels over a shared network, which are accordingly referred to as output channel and input channel.
	Specifically, data transmissions over the output channel occur periodically with interval $\Delta > 0$.
	That is, the output encoder samples $y(t)$ and sends its quantized version to the controller every $\Delta$ time.
	Likewise, the digital controller generates control signals and transmits their quantized values to the plant periodically with interval $\delta>0$.
	At the plant side, the quantized control inputs are first decoded, then pass through a zero-order hold (ZOH) before entering the plant.~To maintain the synchronization between input and output transmissions, we choose $\delta = \Delta/b $ for some $b \in \mathbb{N}$.
	For future reference, let
	\begin{equation*}
	x_{q,k} := x(q\Delta + k\delta),\qquad y_{q,k} := y(q\Delta + k\delta)
	\end{equation*}
	for every $q \in \mathbb{Z}_{\ge 0}$, and $k = 0, \cdots\!, \frac{\Delta}{b}$, and
	\begin{equation}\label{eq:adbd}
	A_d := e^{A\delta}, \qquad B_d := \int_{0}^{\delta}{e^{As}B\, ds}.
	\end{equation}
%	Since $(A, B)$ is controllable, if $\delta$ is non-pathological, then $(A_d, B_d)$ is controllable \cite{Kreisselmeier1999On}, and denote its controllability index by $\eta$.
%	Similarly, $(C, A_d^{\eta})$ is observable.
	Moreover, we use $x_q$ to denote $x_{q,0}$ for simplicity.
	
	\begin{figure}[t]
		\centering
		\includegraphics[width=8cm]{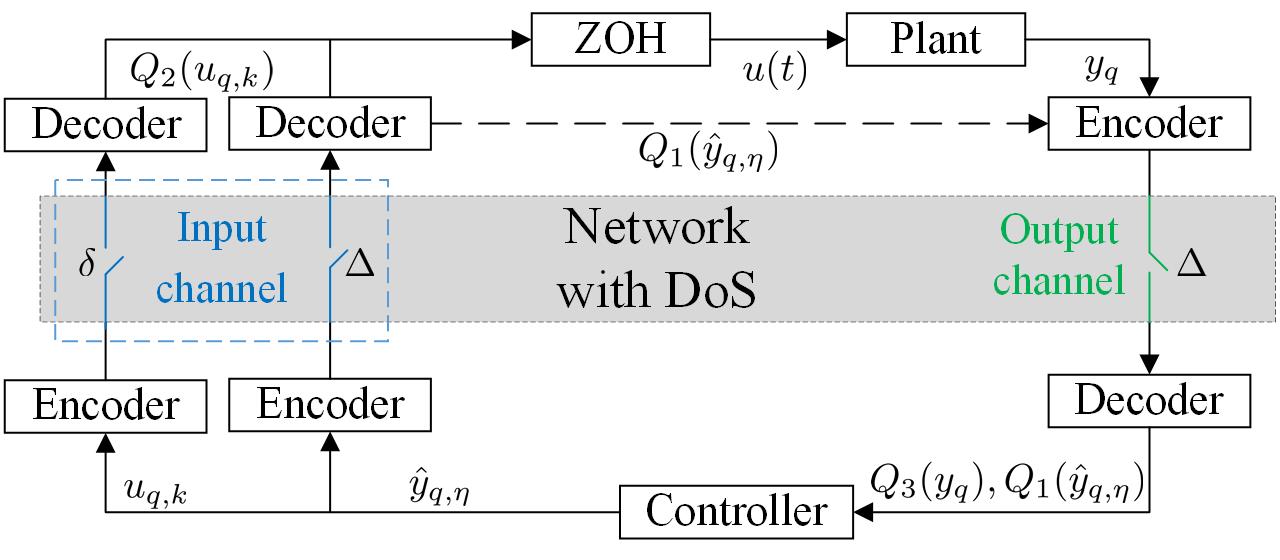}\\
		\caption{Networked control system with both input (the blue line) and output (the green line) channels subject to DoS attacks.}\label{networkfig}
		\centering
	\end{figure}
%	We make the following standard assumption.
%	\begin{assumption}[Initial state bound]\label{x0bound}
%		\textcolor{blue}{The initial state $x_0$ is bounded and known.}

%	\end{assumption}
%\begin{assumption}\label{as:sys}
%	\begin{itemize}
%		\item [A1)]\label{as:abca}
%		The pair $(A,B)$ is assumed controllable, and the pair $(C,A)$ is observable; and,
%		\item [A2)]\label{x0bound}
%		The bound of the initial state $|x_0|$ is   known.
%	\end{itemize}
%\end{assumption}
We make the following assumptions on system \eqref{continuoussystem}.
\begin{assumption}[Controllability and observability]\label{as:abca}
	The pair $(A,B)$ is controllable, and the pair $(C,A)$ is observable.
	%	The pairs $(A,B)$, and $(C,A)$ are assumed controllable, and observable, respectively.
\end{assumption}
\begin{assumption}[Initial state bound]\label{x0bound}
	An upper bound on the initial state $|x_0|$ is known.
\end{assumption}
\begin{remark}
	Thanks to As. \ref{as:abca}, it has been shown in \cite{Kreisselmeier1999On} that if $\delta$ is non-pathological, then $(A_d, B_d)$ in \eqref{eq:adbd} is controllable.
	Let $\eta$ denote its controllability index, which can be computed by evaluating
	${\rm rank} [B_d, \cdots\!, A_d^{\eta} B_d] = n_x$.
	%		This implies that $(A_d, B_d)$ is $\eta$-step controllable, and ${\rm rank} [B_d, \cdots\!, A_d^{\eta} B_d] = n_x$.
	Similarly, $(C, A_d^{\eta})$ is observable.
	An upper bound on the initial state in As. \ref{x0bound} can be derived via the zooming-out method; see \cite[Sec. 4]{8880482}.
\end{remark}

	\subsection{Denial-of-Service attack}
	In Fig. \ref{networkfig}, since both input and output signals are transmitted periodically, we adopt the discrete-time DoS attack model in \cite{8880482}.
	Under this model, attacks are launched only at output transmission instants, and each lasts for an output transmission period $\Delta$.
	This model is general enough since it only poses requirements on the frequency and duration of DoS attacks.
	Here, DoS frequency is the number of DoS \emph{off/on} switches over a fixed time interval, while
	DoS duration represents the total number of attacks.
	\begin{assumption}[DoS frequency]\label{DoS_frequencyassumption}
		There exist constants $\kappa_f \in \mathbb{R}_{\ge 0}$ and $\nu_f \in \mathbb{R}_{\ge 2}$ such that DoS frequency satisfies
		\begin{equation}\label{dosfre}
		\Phi_f(q) \le \kappa_f + \frac{q}{\nu_f}
		\end{equation}
		over time interval $[0, q\Delta)$, where $q \in \mathbb{Z}_{\ge 0}$.
	\end{assumption}
	
	\begin{assumption}[DoS duration]\label{DoS_durationassumption}
		There exist constants $\kappa_d \in \mathbb{R}_{\ge 0}$ and $\nu_d \in \mathbb{Z}_{\ge 1}$ such that DoS duration satisfies
		\begin{equation}\label{dosdur}
		\Phi_d(q) \le \kappa_d + \frac{q}{\nu_d}
		\end{equation}
		over time interval $[0, q\Delta)$, where $q \in \mathbb{Z}_{\ge 0}$.
	\end{assumption}
	\begin{remark}	
		Given its generality, this attack model has been widely used in e.g., \cite{8880482,Feng2020multi,FengResilient,LuInput,feng2020datarate,PersisInput}.
		As pointed out in \cite{Hespanha1999STABILITY}, $\nu_f\Delta$ in As. \ref{DoS_frequencyassumption} can be regarded as the average dwell-time between two consecutive DoS attacks \emph{off/on} switches.
		On the other hand, As. \ref{DoS_durationassumption} indicates that, the average duration of DoS attacks does not exceed a proportion $1/\nu_d$ of the time interval.
		Constants $\kappa_f$ and $\kappa_d$ are also known as chatter bounds.
		Conditions $\nu_d \ge 1$ and $\nu_f \ge 2$ suggest that DoS attacks are not strong enough to prevent all packets from being transmitted, thus rendering it possible for the system to be stabilized by suitable control strategies.
	\end{remark}
	\section{Networked Phenomena at Both Input and Output Channels}\label{inputoutputsection}
	This section aims to design resilient encoding schemes for stabilization of system (\ref{continuoussystem}) via a remote observer-based digital controller over communication channels subject to limited bandwidth and DoS attacks; see Fig. \ref{networkfig}.
	To this end, there are three signals that need to be quantized, i.e., the estimated output by observer $\hat{y}_q$, the control input $u_{q,k}$, and the plant output $y_q$, with their quantized values denote by $Q_1(\hat{y}_q)$, $Q_2(u_{q,k})$, and $Q_3(y_q)$, respectively.
	In addition, since the input and output channels share a communication network, we assume for simplicity that, once there is a DoS attack,
	neither the input nor the output signals will be received, and both of them are set to the default zero.
	In this manner, the decoder and encoder at both input and output sides can infer whether there is an attack.
	Further, their quantization ranges and centers are identical at every transmission instant.
	As a result, they can be synchronized even with an ACK-free protocol.
	\subsection{Controller architecture}
	To stabilize system (\ref{continuoussystem}), we put forth a two-stage observer-based controller by considering whether there is an attack or not.
	Specifically, in the absence of DoS attacks, both $Q_3(y_q)$ and $Q_1(\hat{y}_{q-1,\eta})$ are available at the observer side, so we construct the following controller
	\begin{subequations}\label{abdoscontroller} 
		\begin{align}
		&\hat{x}_{q, k+1} \!=\! A_d\hat{x}_{q,k} \!+\! B_du_{q,k},\!\!  &k &\le \eta - 1 \label{abdoscontroller_1}\\
		&\hat{x}_{q}\!=\! \hat{x}_{q-1,k} \!+\! M\big[Q_3(y_{q}) \!-\! Q_1(\hat{y}_{q-1,k})\big],\!\!  &k &= \eta \label{abdoscontroller_2}\\
		&\hat{y}_{q,k} \!=\! C \hat{x}_{q,k} \label{abdoscontroller_3}\\
		&u_{q,k} \!=\! K \hat{x}_{q,k} \label{abdoscontroller_4}
		\end{align}
	\end{subequations}
	where the initial condition $\hat{x}_0$ is given by $\hat{x}_0 = 0$, and $\delta$ is chosen such that
	\begin{equation}\label{delta}
	\delta = \frac{\Delta}{\eta}.
	\end{equation}
	Matrix $M \in \mathbb{R}^{n_x \times n_y}$ can be regarded as an observer gain such that $R := A_d^{\eta}(I - MC)$ is schur stable, which always exists since $(C, A_d^{\eta})$ is observable.
	Moreover, $(A_d, B_d)$ is controllable, thus a controller gain matrix $K \in \mathbb{R}^{n_u \times n_x}$ can be designed such that
	\begin{equation}\label{dbdb}
	\bar{R}^{\eta} = (A_d + B_dK)^{\eta} = 0.
	\end{equation}
	\begin{remark}\label{remark:dbgain}
		Matrix $K$ obeying (\ref{dbdb}) is also known as a class of deadbeat controller gain, since it assigns all the eigenvalues of $A_d + B_d K$ to the origin.
		Solution of this eigenstructure assignment problem is non-unique, and can be obtained through several approaches, e.g., \cite{FahmyDead}.
	\end{remark}

	On the other hand, when there is a DoS attack, none of $Q_1(\hat{y}_q)$, $Q_2(u_{q,k})$, or $Q_3(y_q)$ can be received, thus we simply employ an open-loop controller as follows
	\begin{subequations}\label{predoscontroller}
		\begin{align}
		&\hat{x}_{q,k+1}  = A_d \hat{x}_{q,k}\\
		&\hat{y}_{q,k}  = C \hat{x}_{q,k}\\
		&u_{q,k}  = 0
		\end{align}
	\end{subequations}
	with the initial estimated state $\hat{x}_0 = 0$.
	
	In addition, to apply the discrete-time signal $Q_2(u_{q,k})$ to the continuous-time system (\ref{continuoussysteG_1}), a ZOH is used, and the control input is given by
	\begin{equation*}
	u(t) = Q_2(u_{q,k}), \qquad q\Delta + k\delta \le t < q\Delta + (k+1)\delta
	\end{equation*}
	where $k = 0, \cdots\!, \eta - 1$.
	\subsection{Quantizer}
	We first design quantizers at the input channel.
	According to (\ref{abdoscontroller_1}) and (\ref{abdoscontroller_2}), $u_{q,k}$ is needed for feedback control, whereas $\hat{y}_{q,k}$, resetting the estimated state, is required at each successful transmission instants.
	Therefore, the controller sends $u_{q,k}$ and $\hat{y}_{q,\eta}$ to the quantizers periodically at a different rate.
	In more precise terms, periods for the former and the latter are $\delta$, and $\eta \delta$, respectively.
	Let $E_{1,q} \ge 0$ and $E_{2,q,k} \ge 0$ satisfy
	\begin{equation}\label{2E12inequality}
	|\hat{y}_{q - 1,\eta}| \le E_{1,q},\ \ \ |u_{q,k}| \le E_{2,q,k}.
	\end{equation}
	Suppose there are $N_1$ ($N_2$) levels for quantization of $\hat{y}_{q,\eta}$ ($u_{q,k}$).
	Partition the hypercubes at the encoders
	\begin{equation*}
	\begin{split}
	\{\hat{y} \in \mathbb{R}^{n_y}: |\hat{y}_{q - 1,\eta}| \le E_{1,q}\},~~
	\{u \in \mathbb{R}^{n_u}: |u_{q,k}| \le E_{2,q,k}\}
	\end{split}
	\end{equation*}
	into $N_1^{n_y}$, and $N_2^{n_u}$ equal-sized boxes, respectively.
	In addition, each box is represented by a value in $\{1, \cdots\!, N_1^{n_y}\}$, or $\{1, \cdots\!, N_2^{n_u}\}$ following a bijection mapping.
	Indices that denote the partitioned boxes containing $\hat{y}_{q, \eta}$ and $u_{q, k}$ are then sent to the decoders.
	If $\hat{y}_{q, \eta}$ and $u_{q, k}$ are on the boundary of several boxes, then anyone of them can be chosen.
	At the decoders side, $Q_1(\hat{y}_{q, \eta})$ and $Q_2(u_{q, k})$ are recovered from the indices.
	This implies that the encoder and its corresponding  decoder should share the same quantization ranges and centers.
	Since DoS attacks block both input and output signals from transmitting, encoders and decoders at both sides of input and output channels are naturally synchronized.
	The quantization errors
	 of the aforementioned encoding schemes obey
	\begin{equation}\label{2E1inequality}
	|\hat{y}_{q - 1,\eta} - Q_1(\hat{y}_{q - 1,\eta})| \le \frac{E_{1,q}}{N_1},
	\end{equation}
	\begin{equation}\label{2E2inequality}
	|u_{q,k} - Q_2(u_{q,k})| \le \frac{E_{2,q,k}}{N_2}.
	\end{equation}
	Since $\hat{x}_0 = 0$, we deduce that $\hat{y}_0 = u_0 = 0$.
	Therefore, the initial bounds $E_{1,0}$ and $E_{2,0,0}$ can be set by
	\begin{equation*}
	E_{1,0} = 0,\ \ \ E_{2,0,0} = 0.
	\end{equation*}
	
	Moreover, as for the output $y_q$, choose $E_{3,q} \ge 0$ such that
	\begin{equation}\label{2E3inequality}
	|y_q - Q_1(\hat{y}_{q - 1,\eta})| \le E_{3,q}.
	\end{equation}
	Let $N_3$ be the quantization level of  $y_q$.
	The hypercube
	\begin{equation}\label{eq:hypercubeQ3}
	\{y \in \mathbb{R}^{n_y} : |y_q - Q_1(\hat{y}_{q - 1,\eta})| \le E_{3, q}\}.
	\end{equation}
	is partitioned into $N_3^{n_y}$ equal-sized boxes with the center $Q_1(\hat{y}_{q - 1, \eta})$.
	Then, following the same procedure as the above two quantizers, $Q_3(y_q)$ is transmitted to the controller every $\Delta$ time.
	The quantization error satisfies
	\begin{equation*}
	|y_q - Q_3(y_q)| \le \frac{E_{3, q}}{N_3}.
	\end{equation*}
	Define error of the system $e_{q, k} := x_{q,k} - \hat{x}_{q,k}$.
	Combining $\hat{x}_0 = 0$ with As. \ref{x0bound}, we deduce that the initial error obeys $|e_0| = |x_0|$. 
	Thus it suffices to set $E_{3,0} := \Vert C\Vert |x_0|$.
	\subsection{Stability analysis}\label{Main_results}
	In this subsection, we start by presenting encoding schemes $\{E_{p,q,k}\} (p = 1, 2, 3)$, followed by formal stability conditions.
	Design $\{E_{1,q}: q \in \mathbb{Z}_{\ge 1}\}$ such that
	\begin{equation}\label{2E1q}
	E_{1,q} = E_{1,0},\qquad \forall q \in \mathbb{Z}_{\ge 1}
	\end{equation}
	and let $\{E_{2,q,k}: q \in \mathbb{Z}_{\ge 1}, k = 0, \cdots\!, \eta - 1\}$ be updated by
	\begin{equation}\label{2E2qk}
	E_{2,q,k} = \frac{N_3 - 1}{N_3}\left\Vert K\bar{R}^kM\right\Vert E_{3,q}, \qquad {\rm{~if~}}q\Delta = s_r.
	\end{equation}
	Moreover, the sequence $\{E_{3,q}: q \in \mathbb{Z}_{\ge 1}\}$ is given by
	\begin{equation}\label{2E3q}
	E_{3,q + 1} :=
	\left\{
	\begin{aligned}
	& \hat{\theta}_a E_{3,q}, & q\Delta \ne s_r\\
	& \hat{\theta}_{0} E_{3,q}, & (q-1)\Delta \ne s_r, q\Delta = s_r\\
	& \hat{\theta}_{na}E_{3,q}, & (q-1)\Delta = s_r, q\Delta = s_r
	\end{aligned}
	\right.
	\end{equation}
	where 
	\begin{align*}
	\hat{\theta}_a & := \Vert A_d^{\eta}\Vert\\
	\hat{\theta}_{0} & := a_0\rho + \frac{\Vert C\Vert a_1}{N_3} + \Vert C\Vert a_2\frac{N_3 - 1}{N_2N_3}\\
	\hat{\theta}_{na} & := \rho + \frac{\Vert C\Vert a_1}{N_3} + \Vert C\Vert a_2\frac{N_3 - 1}{N_2N_3}
	\end{align*}
	with positive constants $a_0, a_1, a_2$, and $0 < \rho < 1$ validating the following for all $\ell \ge 1$
	\begin{equation}\label{2rho}
	\begin{aligned}
	\big\Vert R^{\ell}\big\Vert \le a_0\rho^{\ell},~~~~ \big\Vert R^{\ell}A_d^{\eta}M\big\Vert \le a_1\rho^{\ell}\\
	\sum_{i = 0}^{\eta - 1}{\big\Vert R^{\ell}A_d^{\eta - i - 1}B_d \big\Vert\big\Vert K\bar{R}^iM\big\Vert}\le a_2\rho^{\ell}.\\
	\end{aligned}
	\end{equation}
	Since $R$ is schur stable, there always exist such constants. 
	
	Next, we show that our designed schemes above are resilient to DoS attacks, which is one of our main results too.
	\begin{theorem}\label{2convergetheorem}
		Consider system (\ref{continuoussystem}) with the observer-based controller in (\ref{abdoscontroller}) and (\ref{predoscontroller}), with $K$ obeying (\ref{dbdb}) and $M$ chosen such that $R$ is schur stable.
		Let As. \ref{as:abca}--\ref{DoS_durationassumption} hold.
		If i) the input and output transmission periods adhere to (\ref{delta}), ii) the number of quantization levels $N_1 \ge 1$ is odd, 
		\begin{equation}\label{2Ncondition}
		\begin{split}
		N_2\! >\! \max\!{\left\{\frac{a_2\Vert C\Vert}{1 - \rho},\, \frac{a_2}{a_1}\right\}},\,{\rm{~and~}}
		N_3\! > \!\frac{\Vert C\Vert a_1 - \frac{\Vert C\Vert a_2}{N_2}}{1 - \rho  -\frac{\Vert C\Vert a_2}{N_2}}
		\end{split}
		\end{equation}
		and, iii) DoS attacks satisfy
		\begin{equation}\label{2doscondition}
		\frac{1}{\nu_d} \le \frac{\log{(1/\hat{\theta}_{na})}}{\log{(\hat{\theta}_a/\hat{\theta}_{na})}} - \frac{\log{(\hat{\theta}_0/\hat{\theta}_{na})}}{\log{(\hat{\theta}_a/\hat{\theta}_{na})}}\frac{1}{\nu_f}
		\end{equation}
		then the system is exponentially stable under the encoding scheme with error bounds $\{E_{p,q,k}:q \in \mathbb{Z}_{\ge 1}, k = 0, \cdots\!, \eta - 1\} (p = 1, 2, 3)$ constructed by the update rule in (\ref{2E1q})-(\ref{2E3q}).
	\end{theorem}
	We begin proving Thm. \ref{2convergetheorem} by giving a lemma demonstrating that the update rules in \eqref{2E1q}-\eqref{2E3q} satisfy (\ref{2E12inequality}) and (\ref{2E3inequality}).
	\begin{lemma}\label{lem:lem1}
		Consider system (\ref{continuoussystem}) with the controller in (\ref{abdoscontroller}) and (\ref{predoscontroller}), where $K$ obeys (\ref{dbdb}).
		Let As. \ref{as:abca}--\ref{DoS_durationassumption} hold.
		If $\{E_{p,q,k}:q \in \mathbb{Z}_{\ge 1}, k = 0, \cdots\!, \eta - 1\} (p = 1, 2, 3)$ obey (\ref{2E1q})-(\ref{2E3q}), then (\ref{2E12inequality}) and (\ref{2E3inequality}) hold true for all $q \in \mathbb{Z}_{\ge 1}$.
	\end{lemma}
	\begin{proof}
		Encoding schemes for systems with quantized inputs and outputs in the absence of DoS attacks have been discussed in \cite{WakaikiObserver}.
		However, their methods cannot be directly applied due to the DoS-induced coupling between these schemes.
		This challenge is addressed through our carefully designed controller structure in (\ref{abdoscontroller})-(\ref{dbdb}).
		According to (\ref{abdoscontroller}) and (\ref{dbdb}), 
		\begin{equation*}
		\hat{y}_{q - 1, \eta} = C\hat{x}_{q - 1, \eta} = C\bar{R}^{\eta}\hat{x}_{q - 1} = 0
		\end{equation*}
		holds true irrespective of DoS attacks, 
		which implies $E_{1,q} \ge |\hat{y}_{q - 1, \eta}| = E_{1,0}$ for all  $q \in \mathbb{Z}_{\ge 1}$, so $E_{1,q}$ remains unchanged.
		This result further indicates that $Q_1(\hat{y}_{q - 1, \eta}) = 0$.
		Hence, it follows from (\ref{eq:hypercubeQ3}) that the quantization center of $Q_3(y_q)$ is at the origin.
		On the other hand, if no DoS attacks occur within $[q_1\Delta, (q_1 + 1)\Delta)$, then
		\begin{equation}\label{2hatxqk}
		\begin{aligned}
		\hat{x}_{q_1, k} =&\ (A_d + B_dK)^{k }\hat{x}_{q_1} \\
		= &\ \bar{R}^{k}(\hat{x}_{q_1 - 1, \eta} + M[Q_3(y_{q_1}) - Q_1(\hat{y}_{q_1 - 1, \eta})])\\
		= &\  \bar{R}^k M[Q_3(y_{q_1}) - Q_1(\hat{y}_{q_1 - 1, \eta})]
		\end{aligned}
		\end{equation}
		hence $u_{q_1,k}$ can be expressed by $Q_3(y_{q_1}) - Q_1(\hat{y}_{q_1 - 1, \eta})$.
		In addition, since
		\begin{equation*}
		|Q_3(y_q) - Q_1(\hat{y}_{q-1, \eta})| \le \frac{N_3 - 1}{N_3}E_{3,q}
		\end{equation*}
		it follows that, in the absence of DoS attacks, $u_{q,k}$ satisfies
		\begin{equation}\label{2uqk}
		|u_{q,k}| \le \frac{N_3 - 1}{N_3}\big\Vert K\bar{R}^kM\big\Vert E_{3,q} =: E_{2,q,k}
		\end{equation}
		for all $q \ge 1, k = 0, \cdots\!, \eta -1$.
		When DoS attacks occur, the plant cannot receive inputs from controller.
		In other words, $E_{2,q,k}$ only depends on the latest $E_{3,q}$, thus $E_{2,q,k}$ can remain unchanged during DoS attacks.
		
		Following the definitions of $E_{1,q}$ and $E_{2,q,k}$, we are able to design sequence $\{E_{3,q}: q\in \mathbb{Z}_{\ge 1}\}$.
		First, in the absence of DoS attacks, the error just before each transmission instant, $e_{q - 1, \eta} = x_q - \hat{x}_{q - 1,\eta}$, satisfies
		\begin{align*}
		e_{q - 1, \eta} =&\ A_d^{\eta}(I - MC)e_{q - 1} - A_d^{\eta}M[Q_3(y_q)- y_q]\\
		& + \sum_{i = 0}^{\eta - 1}{A_d^{\eta - i - 1}B_d[Q_2(u_{q-1,i}) - u_{q- 1,i}]}\\
		& - A_d^{\eta}\big[\hat{y}_{q - 1, \eta} - Q_1(\hat{y}_{q - 1, \eta})\big]
		\end{align*}
		which implies that $e_{q - 1, \eta}$ generally relies on $\hat{y}_{q-1, \eta}, u_{q,k}$, and itself, thus introducing coupling in $E_{3,q}$ design.
		Here, this issue is addressed by (\ref{dbdb}).
		To see this, recalling (\ref{2uqk}), $\hat{y}_{q - 1, \eta} = 0$, and $Q_1(\hat{y}_{q - 1, \eta}) = 0$, we have that
		\begin{align}\label{2error}
		&e_{q + {\ell} - 1, \eta} =R^{\ell} e_{q - 1} + \sum_{j = 0}^{{\ell} - 1}R^jA_d^{\eta}M(Q_3(y_{q - j})- y_{q - j})\nonumber\\
		&+ \sum_{j = 0}^{{\ell} - 1}R^j\sum_{i = 0}^{\eta \!-\! 1}A_d^{\eta \!-\! i - 1}B_d \big[Q_2(u_{q+\ell \!-\!j \!-\!1,i}) \!-\! u_{q+{\ell}-j-1,i}\big].
		\end{align}
		Define $E_{3,q}$ as follows
		\begin{align*}
		E_{3,q+{\ell}} :=&\ a_0\rho^{\ell}E_{3,q} \\
		& + \sum_{i = 0}^{{\ell} - 1}\left(\frac{(N_3 - 1)a_2\Vert C\Vert }{N_2N_3} + \frac{\Vert C\Vert a_1}{N_3}\right)\rho^{i}E_{3,q-i}.
		\end{align*}
		Hence, combining (\ref{2uqk}) with (\ref{2error}) yields
		\begin{align}\label{eq:yqy}
		|y_{q + 1} - Q_1(\hat{y}_{q, \eta})| & \le |y_{q + 1} - \hat{y}_{q, \eta}| + |\hat{y}_{q, \eta} - Q_1(\hat{y}_{q, \eta})|\nonumber\\
		& \le \Vert C\Vert |x_{q + 1} - \hat{x}_{q, \eta}|\nonumber\\
		& \le \hat{\theta}_{na} E_{3,q}
		\le E_{3,q + 1}.
		\end{align}
		
		Moreover, since both the input and output channels are blocked in the presence of DoS attacks, and $\hat{y}_{q,\eta} = 0$, due to the property of $\bar{R}$, it follows that
		\begin{equation*}
		|y_{q + 1} - Q_1(\hat{y}_{q, \eta})| \le \hat{\theta}_{a} E_{3,q} \le E_{3,q + 1}
		\end{equation*}
		and we complete the proof.
	\end{proof}
	Next, we establish upper bounds on the sequences $\{E_{p,q,k}:q \in \mathbb{Z}_{\ge 1}, k = 0, \cdots\!, \eta - 1\} (p = 1, 2, 3)$, whose existence will imply the boundness of state trajectory. 
	\begin{lemma}\label{lem:eqbound}
		Consider system (\ref{continuoussystem}) with controller in (\ref{abdoscontroller}) and (\ref{predoscontroller}), where $K$ satisfies (\ref{dbdb}) and $M$ is chosen such that $R$ is schur stable.
		Let the assumptions and conditions in Thm. \ref{2convergetheorem} hold.
		If further $\{E_{p,q,k}:q \in \mathbb{Z}_{\ge 1}, k = 0, \cdots\!, \eta - 1\} (p = 1, 2, 3)$ obey (\ref{2E1q})-(\ref{2E3q}),  there exist $\Omega_1 \ge 1$ and $\gamma \in (0,1)$ such that
		\begin{equation}\label{Eqomegagamma}
		E_{3, q} \le \Omega_1 \gamma^q |x_0| ,\qquad \forall k \in \mathbb{Z}_{\ge 1}
		\end{equation}
		and $E_{1, q} = E_{1, 0}$, and $E_{2, q, k} \le \Omega_2\gamma^q|x_0|$.
	\end{lemma}
	\begin{proof}
		Using (\ref{2E1q}), $E_{1, q}$ remains unchanged within the considered interval, therefore, $E_{1, q} = E_{1, 0}$ holds for all $q \in \mathbb{Z}_{\ge 1}$.
		The proof for $E_{3, q} \le \Omega \gamma^q |x_0|, \forall q \in \mathbb{Z}_{\ge 1}$ follows directly from that of Lemma
		3.9 in \cite{8880482}, 
		where $\Omega_1 := \frac{\hat{\theta}_0^{\Pi_f + 1}\hat{\theta}_a^{\Pi_d}}{\hat{\theta}_{na}^{\Pi_f + \Pi_d + 1}} \big(\hat{\theta}_{na}\big(\frac{\hat{\theta}_0}{\hat{\theta}_{na}}\big)^{\nu_f}\big(\frac{\hat{\theta}_a}{\hat{\theta}_{na}}\big)^{\nu_d}\big)$.
		Moreover, applying (\ref{2uqk}), $E_{2, q, k} \le \Omega_2\gamma^q |x_0|$ can be verified with $\Omega_2 := \frac{N_3 - 1}{N_3}\Vert K\bar{R}^{\eta - 1}\Vert\Omega_1$.
	\end{proof}
	We are now in a position to prove Thm. \ref{2convergetheorem}.
	\begin{proof}
		[Proof of Theorem \ref{2convergetheorem}]
		We first establish the bound of the state $x$ at the transmission instants, i.e., $|x(q\Delta)|$, then derive its bound at the sampling instants, i.e.,  $|x(q\Delta + k\delta)|$. 
		Finally, combining these two bounds to yield bound $|x(t)|$ in the considered horizon.
		
		First, according to (\ref{continuoussystem}), (\ref{abdoscontroller}), and  (\ref{predoscontroller}), one has
		\begin{align}\label{1xqeta2}
		x_{q, \eta} =&~ \bar{R}^{\eta}x_{q,k} + \sum_{i = 0}^{\eta - 1}\bar{R}^iB_dK(x_{q, \eta - i - 1} - \hat{x}_{q, \eta - i - 1}) \nonumber\\
		& + \sum_{i = 0}^{\eta - 1}\bar{R}^iB_d(Q_2(u_{q, \eta - i - 1}) - K \hat{x}_{q, \eta - i - 1})
		\end{align}
		and
		\begin{align}\label{eq:vertx}
		\vert x_{q, \eta}\vert \le&~ \Vert \bar{R}^{\eta}\Vert |x_{q}| + \sum_{i = 0}^{\eta - 1}\Vert\bar{R}^iB_dK\Vert |(x_{q, \eta - i - 1} - \hat{x}_{q, \eta - i - 1})|\nonumber\\
		& + \sum_{i = 0}^{\eta - 1}\Vert \bar{R}^iB_d\Vert|(Q_2(u_{q, \eta - i - 1}) - K \hat{x}_{q, \eta - i - 1})|.
		\end{align}
		Since (\ref{eq:yqy}), it follows that
		\begin{equation}\label{eq:xqx}
		\Vert x_{q} - \hat{x}_{q - 1, \eta}\Vert \le \frac{E_{3,q}}{\Vert C\Vert}.
		\end{equation}
		Noticing that $\bar{R}^{\eta} = 0$, substituting (\ref{2E2qk}) and (\ref{eq:xqx}) into (\ref{eq:vertx}),
		\begin{align}
		\vert x_{q, \eta}\vert\le&\  \sum_{i = 0}^{\eta - 1} \frac{N_3 - 1}{N_2 N_3}\Vert \bar{R}^i B_d\Vert\Vert K\bar{R}^{\eta - i - 1}M\Vert E_{3,q}\nonumber\\
		& + \sum_{i = 0}^{\eta - 1} \frac{\Vert \bar{R}^i B_d K A_d^{\eta - i - 1}\Vert}{\Vert C\Vert} E_{3,q}\nonumber\\
		\le&\ \Omega_x \Omega_1 \gamma^q |x_0|
		\end{align}
		where $\Omega_x := \sum_{i = 0}^{\eta - 1}\big\{ \frac{N_3 - 1}{N_2 N_3}\Vert \bar{R}^i B_d\Vert\Vert K\bar{R}^{\eta - i - 1}M\Vert  + \frac{\Vert \bar{R}^i B_d K A_d^{\eta - i - 1}\Vert}{\Vert C\Vert}\big\}$, and the last inequality holds due to (\ref{Eqomegagamma}).
		
		Since $x_{q, k+1} = A_dx_{q,k} + B_d Q_2(u_{q,k})$, we have that
		\begin{align}\label{eq:xql}
		\vert x_{q, {\ell}}\vert\le &\  \Vert\bar{R}^{{\ell}}\Vert |x_{q}| + \sum_{i = 0}^{{\ell} - 1} \frac{\Vert \bar{R}^i B_d K A_d^{{\ell} - i - 1}\Vert}{\Vert C\Vert}E_{3,q}\nonumber\\
		&+ \sum_{i = 0}^{{\ell} - 1} \frac{N_3 - 1}{N_2 N_3}\Vert \bar{R}^i B_d\Vert\Vert K\bar{R}^{{\ell} - i - 1}M\Vert E_{3,q}\\
		\le&\ \Omega_x \Omega_1 \gamma^q |x_0| + \Omega_3 \gamma^q |x_0|\nonumber \le \bar{\Omega}_x \gamma^q |x_0|
		\end{align}
		where $\Omega_3 := \Omega_1\max_{{\ell}} \sum_{i = 0}^{{\ell}}\big\{ \frac{N_3 - 1}{N_2 N_3}\Vert \bar{R}^i B_d\Vert\Vert K\bar{R}^{{\ell} - 1}M\Vert  + \frac{\Vert \bar{R}^i B_d K A_d^{\eta - i - 1}\Vert}{\Vert C\Vert}\big\}, {\ell} \in \{1, \cdots\!, \eta - 1\}$,  and $\bar{\Omega}_x := \Omega_3 + \Vert\bar{R}^{{\ell}}\Vert\Omega_x$.
		
		Finally, abiding by (\ref{continuoussystem}), $x(t)$ satisfies
		\begin{align*}
		x(t) = e^{A(t - q\Delta - k\delta)} + \int_{q\Delta + k\delta}^te^{As}BQ_2(u_{q, k}) \,ds
		\end{align*}
		for all $t \in [q\Delta + k\delta, q\Delta + (k+1)\delta)$.
		Combining Lem. \ref{lem:eqbound} and (\ref{eq:xql}), it follows that
		\begin{align*}
		\vert x(t) \vert \le \big(\Vert A_d\Vert \bar{\Omega}_x + \frac{N_2 + 1}{N_2}\Vert B_d\Vert\Omega_2\big)\gamma^q|x_0|\le \tilde{\Omega}_x e^{-\sigma t}|x_0|
		\end{align*}
		where $\sigma := \frac{1}{\eta \delta}\log \frac{1}{\gamma}$ and $\tilde{\Omega}_x := \Vert A_d\Vert \bar{\Omega}_x + \frac{N_2 + 1}{N_2}\Vert B_d\Vert\Omega_2$.
		This implies exponential convergence of the state.
	\end{proof}
	\begin{remark}\label{2dbdbremark}
		Leveraging the same technique as in Rmk. \ref{remark:dbgain}, one can also design $M$ to nullify $R^{\mu} = 0$, where $\mu$ is the observability index of $(C, A_d^{\eta})$.
		A direct benefit from using the deadbeat observer gain is that the encoding schemes can be simplified, since $R^{\ell} = 0$ holds for all $\ell \ge \mu$.  
		However, the results in \cite{WakaikiObserver} indicate that despite exhibiting faster convergence and fewer quantization levels, due to the deadbeat property of matrices $R$ and $\bar{R}$, the quantization step size $E_{p,q}/N_p$ is large, which leads to large quantization errors.
		Moreover, it was shown in \cite{8880482} that if the quantization step size $E_{p,q}/N_p$ grows slower during DoS attacks, then the overshoot from an attack is smaller, and the level of system robustness is stronger.
		Therefore, instead of a deadbeat observer gain, a general one that can make $R$ schur stable is employed in the present work.
	\end{remark}
	\section{Network Phenomena at Output Channel}\label{outputsection}
	In this section, we consider stabilizing linear systems over a communication network, where only the output channel is subject to DoS attacks, i.e., the input channel is assumed ideal; see Fig. \ref{siglenetworkfig}.
	The transmission policy in the previous section is considered here; that is, the digital controller receives quantized output $Q(y_q)$ from the plant with period $\Delta$ and generates control input $u_{q,k}$ with period $\delta$.
	Notice that the decoder can recover the correct quantized value from the index sent by the encoder only if they share the same quantization ranges and centers.
	It is thus necessary to ensure that the encoder and the decoder are \emph{synchronized} before designing encoding schemes.
	A direct way to maintain synchronization is through using an ACK-based protocol; see Fig. \ref{tcpfig}, which has been adopted in previous studies, such as, \cite{feng2020datarate, 8880482}.
	Nevertheless, in real-time applications, protocols without ACKs, e.g., UDP, are often preferred since the resulting implementation is simpler as well as saves the additional energy required for sending ACKs \cite{HongUDP}.
	Hence, in the following, we first show that method for stabilizing systems with ACK-based protocols can no longer be used under ACK-free protocols.
	Then, we demonstrate that our proposed methods can inform the encoder of DoS attacks from zero inputs, thus the decoder and the encoder can be synchronized even without ACKs.
	\begin{figure}
		\centering
		\includegraphics[width=6cm]{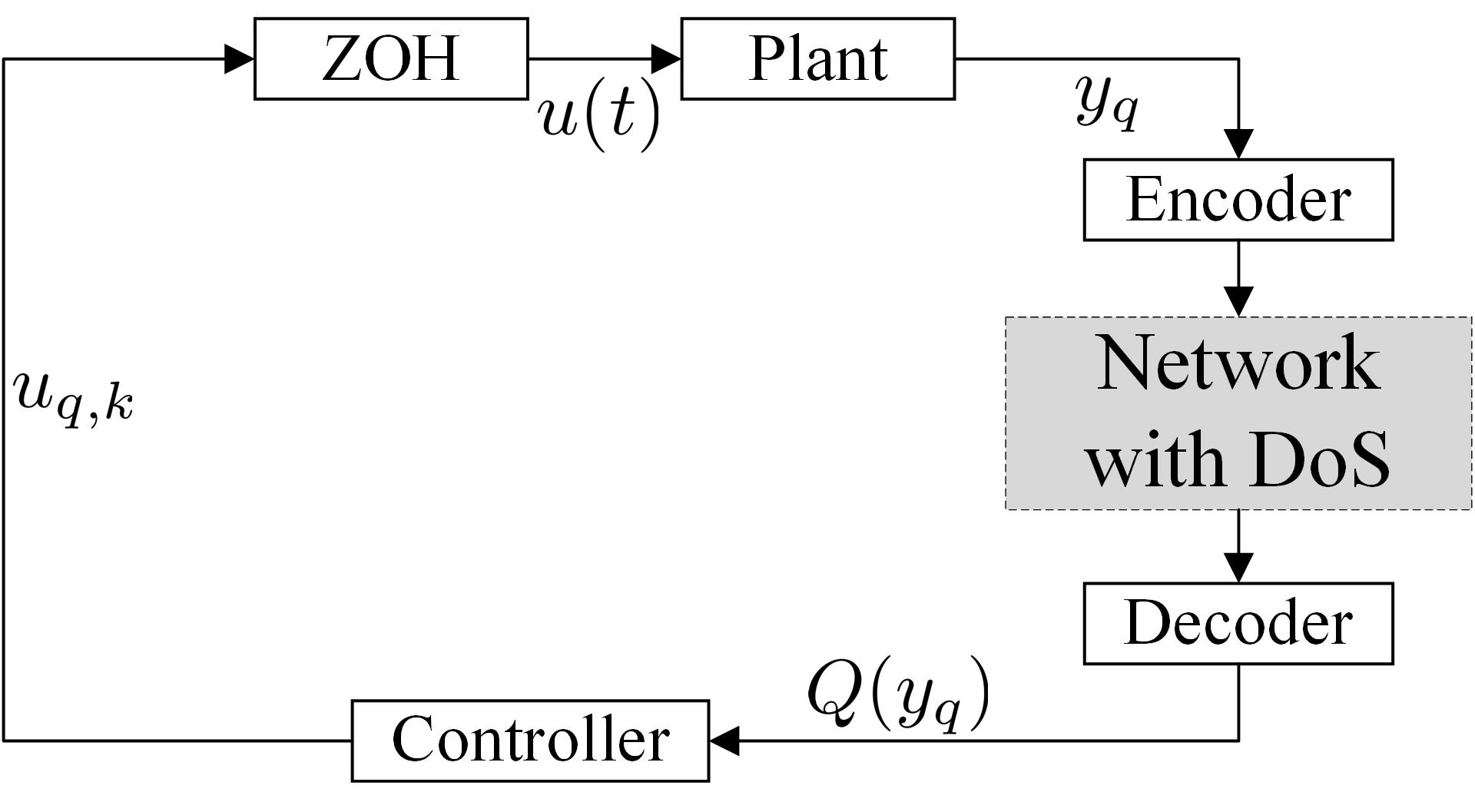}\\
		\caption{Closed-loop system 
			with an ACK-free protocol.}\label{siglenetworkfig}
		\centering
	\end{figure}
	\begin{figure}
		\centering
		\includegraphics[width=6cm]{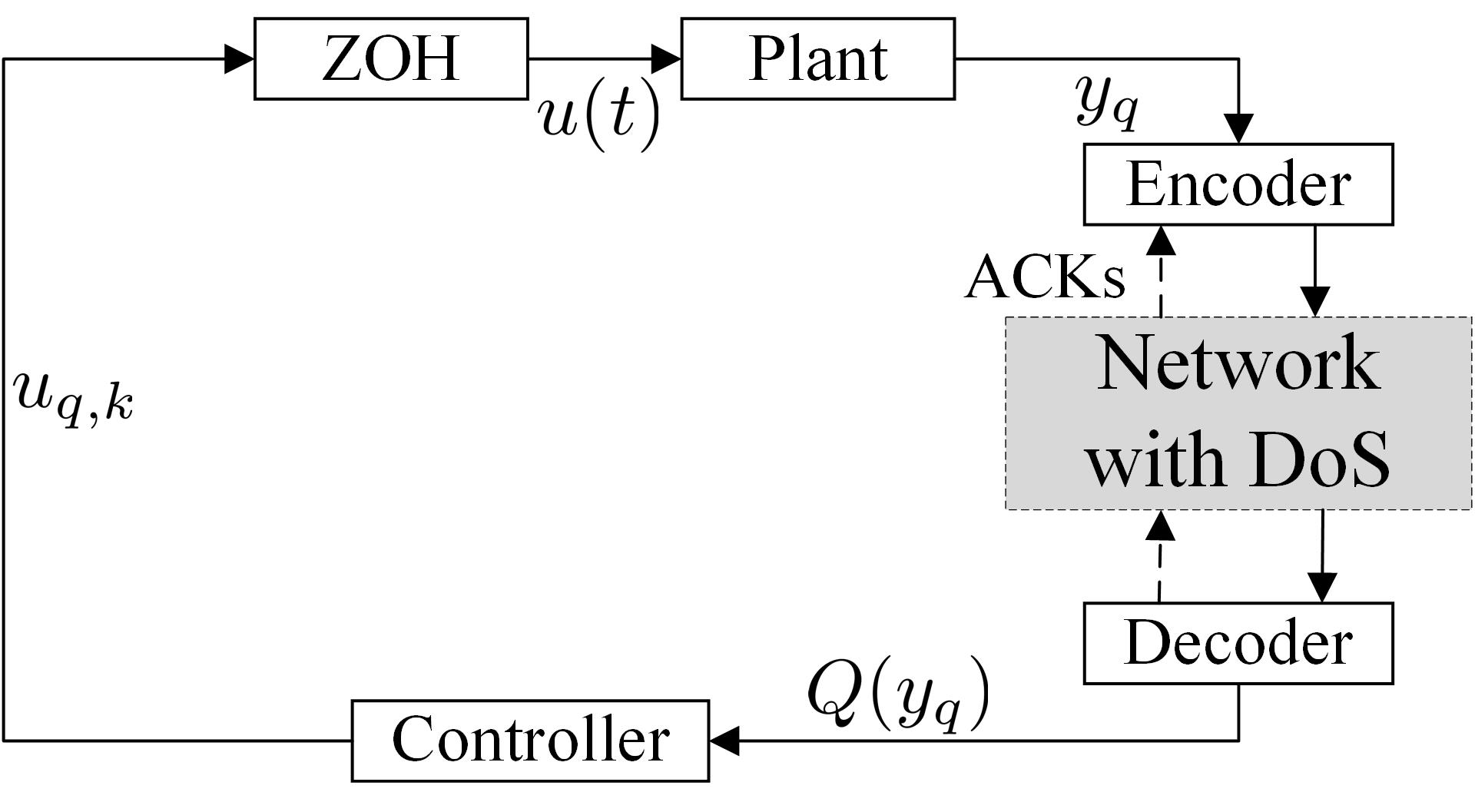}\\
		\caption{Closed-loop system 
			with an ACK-based protocol. The black dashed line represents the ACKs sent from the decoder to the encoder.}\label{tcpfig}
		\centering
	\end{figure}
	\subsection{Controller under an acknowledgment-based protocol}
	Recall that $\{s_r\}_{r\in \mathbb{N}_0}$ collects the sequence of successful transmission instants.
	Let $\delta = \Delta$, and choose $K$ such that $\bar{R} = A_d + B_dK$ is schur stable.
	
	We consider an observer-based controller described by
	\begin{subequations}\label{eq:tcpcontroller}
		\begin{align}
		&\hat{x}_{q+1} = A_d\hat{x}_q + B_du_q + L(Q(y_q) - \hat{y}_q), & q\Delta = s_r \label{eq:tcpcontroller_1}\\
		&\hat{x}_{q+1} = A_d\hat{x}_q + B_du_q,  & q\Delta \ne s_r \label{eq:tcpcontroller_2}\\
		&\hat{y}_q = C\hat{x}_q \label{eq:tcpcontroller_3}\\
		&u_q = K\hat{x}_q \label{eq:tcpcontroller_4}
		\end{align}
	\end{subequations}
	where $\hat{x}_{q}\! \in\! \mathbb{R}^{n_x}, \hat{y}_{q} \!\in\! \mathbb{R}^{n_y}$, and $Q(y_{q}) \!\in\! \mathbb{R}^{n_y}$ are the estimated state, the estimated output, and the quantized output, respectively.
	The initial condition is set to be $\hat{x}_0 = 0$.
	Since the input channel is ideal, it follows that
	\begin{equation*}
	u(t) = u_q, \quad q \Delta \le t < (q + 1)\Delta, \quad q \in \mathbb{Z}_{\ge 0}.
	\end{equation*}
	To design an encoding scheme such that the output $y_q$ can be quantized without saturation, an error bound between the estimated output and the actual output, i.e., $| e_q| := |x_q - \hat{x}_q| \le E_q$, should be derived.
	Based on (\ref{continuoussystem_2}) and (\ref{eq:tcpcontroller_2}), it can be deduced that
	\begin{equation}\label{eq:tcperror}
	|y_q - \hat{y}_q| = |C(x_q - \hat{x}_q)| = |Ce_q| \le \Vert C \Vert E_q.
	\end{equation}
	Let $N$ denote the number of quantization levels of $y_q$.
	Similar to the previous section, we partition the hypercube $	\{ y \in \mathbb{R} ^{n_y} : | y_q - \hat{y}_q| \le \Vert C \Vert E_q\}$
	into $N^{n_y}$ equal-sized boxes.
	The quantization error obeys $	|Q(y_q) - y_q| \le \frac{\Vert C\Vert}{N}E_q$.
	According to As. \ref{x0bound}, the initial value $E_0$ is given by
	\begin{equation}\label{eq:tcpinitial}
	|e_0| = |x_0| \stackrel{\triangle}{=} E_0.
	\end{equation}
	Sequence $\{E_q, q \in \mathbb{Z}_{\ge 1}\}$ will be specified latter.
	Notice that the hypercube center is $\hat{y}_q$, which is generated by the predictor-based observer in (\ref{eq:tcpcontroller}).
	Therefore, this predictor should also be equipped at the encoder side.
	Under ACK-based protocol, the decoder sends ACKs to the encoder without delay at successful transmission instants; and when the encoder does not receive the ACKs, it infers that there is a DoS attack. 
	In this manner, synchronization between these two predictors is ensured, which consequently implies that the quantization ranges and the centers at the encoder are identical to that of the decoder.
	
	Before giving stability condition for ACK-based protocol case, we present an output encoding scheme.
	Let 
	\begin{equation}\label{eq:tcperrorbound}
	E_{q + 1} :=
	\left\{
	\begin{aligned}
	& \theta_a E_{q}, & q\Delta \ne s_r\\
	& \theta_0 E_{q}, & (q-1)\Delta \ne s_r, q\Delta = s_r\\
	& \theta_{na}E_{q}, & (q-1)\Delta = s_r, q\Delta = s_r
	\end{aligned}
	\right.
	\end{equation}
	with 
	\begin{subequations}\label{eq:tcpencode}
		\begin{align}
		\theta_{a} &:=\left\|A_d \right\| \label{eq:tcpencode1}\\
		\theta_{0} &:=H_{0} \rho+\frac{H_1\left\|C \right\|}{N} \label{eq:tcpencode2}\\
		\theta_{na} &:=\rho+\frac{H_1 \left\|C \right\|}{N} \label{eq:tcpencode3}
		\end{align}
	\end{subequations}
	where constants $H_0$, $H_1$, and $0 < \rho < 1$ satisfy
	\begin{equation*}
	\Vert (A_d - LC)^\ell \Vert \le H_0 \rho^{\ell},\quad \Vert (A_d - LC)^\ell L\Vert \le H_1 \rho^{\ell}.
	\end{equation*}
	
	\begin{theorem}\label{th:outputtheorem1}
		Consider system (\ref{continuoussystem}) with controller (\ref{eq:tcpcontroller}), where $M$ and $K$ are chosen such that $A_d - LC$ and $A_d + B_d K$ are schur stable.
		Under As. \ref{as:abca}--\ref{DoS_durationassumption}, if i) the quantization levels 
		\begin{equation}\label{2Ntcpcondition}
		\begin{split}
		N > \frac{H_1\Vert C\Vert}{1 - \rho}
		\end{split}
		\end{equation}
		and, ii) DoS attacks satisfy
		\begin{equation}\label{2dostcpcondition}
		\frac{1}{\nu_d} \le \frac{\log{(1/{\theta}_{na})}}{\log{({\theta}_a/{\theta}_{na})}} - \frac{\log{({\theta}_0/{\theta}_{na})}}{\log{(\theta_a/{\theta}_{na})}}\frac{1}{\nu_f}
		\end{equation}
		then the system is exponentially stable under the encoding scheme with error bounds $\{E_{q}:q \in \mathbb{Z}_{\ge 1}\}$ constructed by the update rule (\ref{eq:tcperrorbound}).
	\end{theorem}
	The proof is similar to that of \cite[Thm. 3.4]{8880482} and is thus omitted here due to space limitations.
	\subsection{Controller under an acknowledgment-free protocol}\label{1Acontrollersection}
	In this subsection, we show that the aforementioned controller and encoding scheme cannot stabilize the system when the ACK-based protocol is replaced by an ACK-free protocol.
	This is because synchronization between the encoder and decoder is no longer guaranteed.
	To see this, consider controller (\ref{eq:tcpcontroller}) with the encoding scheme in (\ref{eq:tcperrorbound}) employing an ACK-free protocol.
	In this setting, predictors at the encoder and decoder sides may become asynchronized, since no matter whether DoS attacks happen or not, the decoder does not send ACKs to the encoder.
	When a DoS attack occurs, the predictor at the controller side switches to (\ref{eq:tcpcontroller_2}), whereas the predictor at the encoder side sticks to (\ref{eq:tcpcontroller_1}).
	Moreover, the update rule of sequence $E_q$ at the decoder switches to (\ref{eq:tcpencode1}), while adhering to (\ref{eq:tcpencode2})-(\ref{eq:tcpencode3}) at the encoder.
	As a result, their quantization ranges and centers may deviate, and the correct output value cannot be recovered by the decoder.
	We prove that even if one DoS attack occurs (i.e.,  decoder and encoder are asynchronized for only one transmission period), the state may diverge eventually.
	
	To distinguish between predictors at the encoder and decoder, let $\hat{x}_q$, $\hat{y}_q$, and $\hat{Q}(y_q)$ denote the estimated state, estimated output, and quantized output at the controller side, and $\tilde{x}_q, \tilde{y}_q$, and $Q(y_q)$ denote their counterparts at the encoder side.
	In addition, let $u_q$ stand for the input sent by the controller, and $\tilde{u}_q$ the estimated input generated by the predictor at the encoder side. 
	Predictor at the controller side can be expressed by
	\begin{subequations}\label{eq:tcpdecoder}
		\begin{align}
		&\hat{x}_{q+1} = A_d\hat{x}_q + B_du_q + L(\hat{Q}(y_q) - \hat{y}_q), & q\Delta = s_r \label{eq:tcpdecoder_1}\\
		&\hat{x}_{q+1} = A_d\hat{x}_q + B_du_q,  & q\Delta \ne s_r \label{eq:tcpdecoder_2}\\
		&\hat{y}_q = C\hat{x}_q \label{eq:tcpdecoder_3}\\
		&u_q = K\hat{x}_q \label{eq:tcpdecoder_4}
		\end{align}
	\end{subequations}
	and predictor at the encoder side is described by
	\begin{subequations}\label{eq:tcpencoder}
		\begin{align}
		&\tilde{x}_{q+1} = A_d\tilde{x}_q + B_d\tilde{u}_q + L({Q}(y_q) - \tilde{y}_q)\label{eq:tcpencoder_1}\\
		&\tilde{y}_q = C\tilde{x}_q \label{eq:tcpencoder_2}\\
		&\tilde{u}_q = K\tilde{x}_q
		\end{align}
	\end{subequations}
	where $q \in \mathbb{Z}_{\ge 0}$.
	Similarly, let $E_{d,q}$, and $E_{e,q}$ denote the error bound at the decoder, and the encoder side, respectively
	\begin{equation*}
	\begin{aligned}
	&E_{d,q + 1} :=
	\left\{
	\begin{aligned}
	& \theta_a E_{d,q}, & q\Delta \ne s_r\\
	& \theta_0 E_{d,q}, & (q-1)\Delta \ne s_r, q\Delta = s_r\\
	& \theta_{na}E_{d,q}, & (q-1)\Delta = s_r, q\Delta = s_r
	\end{aligned}
	\right.\\
	&E_{e,q + 1} :=
	\left\{
	\begin{aligned}
	& \theta_0 E_{e,q}, & q\Delta = 0\\
	& \theta_{na}E_{e,q}, & q\Delta > 0
	\end{aligned}
	\right.
	\end{aligned}
	\end{equation*}
	where $\theta_{0}$, $\theta_{a}$, and $\theta_{na}$ are defined in (\ref{eq:tcpencode}).
	Accordingly, the errors at the encoder and decoder sides are $e_{e,q} := x_q - \tilde{x}_q$, and $e_{d,q} := x_q - \hat{x}_q$.
	Moreover, the quantized outputs in (\ref{eq:tcpdecoder_2}) and (\ref{eq:tcpencoder_1}) are
	\begin{align}\label{hatQ}
	\hat{Q}(y_q) = \hat{y}_q + Q^{i}_q\frac{\Vert C\Vert E_{d,q}}{N}\\
	Q(y_q) = \tilde{y}_q + Q^{i}_q\frac{\Vert C\Vert E_{e,q}}{N}
	\end{align}
	where $Q_q^{i}$ denotes the quantization index transmitted from the encoder to the decoder.

	Suppose that a DoS attack is launched at $q_a\Delta$ and no attacks happen before or after $q_a\Delta$. It follows that $\hat{Q}(y_q) = Q(y_q)$ for all $q \le q_a$, and
	\begin{subequations}\label{tildex-x}
		\begin{align}
		&\hat{x}_{q_a} = \tilde{x}_{q_a} \\
		&\hat{x}_{q_a\!+\!1} = (A_d \!+\! B_d K)\hat{x}_{q_a} \\
		&\tilde{x}_{q_a \!+\! 1} = (A_d \!+\! B_d K) \tilde{x}_{q_a} \!+\! L Q^{i}_{q_a}\frac{\Vert C\Vert E_{e,q_a}}{N} \\
		&\hat{x}_{q_a+2}  = (A_d \!+\! B_d K)^2 \hat{x}_{q_a} \!+\! L Q^{i}_{q_a\!+\!1}\frac{\Vert C\Vert E_{d,q_a \!+\! 1}}{N}\\
		\begin{split}
		&\tilde{x}_{q_a + 2} =
		(A_d \!+\! B_d K)^2 \tilde{x}_{q_a} \!+\! L Q^{i}_{q_a+1}\frac{\Vert C\Vert E_{e,q_a + 1}}{N}\\
		&~~~~~~~~~~ \!+\! (A_d \!+\! B_d K)L Q^{i}_{q_a}\frac{\Vert C\Vert E_{e,q_a}}{N} \\
		\end{split}\\
		& \cdots \nonumber
		\end{align}	
	\end{subequations}
	Notice that the quantizer operates normally without saturation only if $E_{e,q} \ge |e_{e,q}| = |x_{q} - \tilde{x}_q|$ and $E_{d,q} \ge |e_{d,q}| = |x_{q} - \hat{x}_q|$ hold for all $q \in \mathbb{Z}_{\ge 1}$.
	If the quantizer saturates, the error between the actual output and the quantized output maybe large, which consequently renders the system unstable.
	In the following, we assume that the quantizer is not saturated; that is $E_{e,q} \ge |e_{e,q}|$ and $E_{d,q} \ge |e_{d,q}|$ for all $q \in \mathbb{Z}_{\ge 1}$, and reach a contradiction.
	Since $E_{e,q + 1} = \theta_{na}E_{e,q}, q > 0$, and $\theta_{na} <1$, sequence $\{|x_q - \tilde{x}_q|\}$ is decreasing.
	Let $\Pi_{L} := A_d - LC$ and $\Pi_{K} := A_d + B_dK$.
	Combining (\ref{hatQ}) and (\ref{tildex-x}) yields
	\begin{align*}
	|x_{q_a + 1} - \tilde{x}_{q_a + 1}| & = |\Pi_{L} (x_{q_a} - \tilde{x}_{q_a}) - L(Q(y_{q_a}) - y_{q_a})|\\
	& \le E_{e, q_a + 1} \stackrel{\triangle}{=} \tilde{E}_{e,q_a + 1}.
	\end{align*}
	Likewise,
	\begin{align*}
	&~|x_{q_a + 2}- \tilde{x}_{q_a + 2}| \\
	=&\ \big|\Pi_{L} (x_{q_a + 1} - \tilde{x}_{q_a + 1}) - L(Q(y_{q_a + 1}) - y_{q_a + 1})\\
	& - BKLQ^{i}_{q_a}\frac{\Vert CR^{-1}\Vert E_{e,q_a}}{N}\big|\\
	\le &\ E_{q_a + 2}  + \frac{1}{\theta_{na}^2}\frac{\Vert BKLQ^{i}_{q_a}\Vert\Vert CR^{-1}\Vert}{N}E_{e,q_a+2}\\
	\stackrel{\triangle}{=} &\ \tilde{E}_{e,q_a + 2}.
	\end{align*}
	Iteratively, for $\ell \ge 3$, it follows that
	\begin{align*}
	&~ |x_{q_a + \ell} - \tilde{x}_{q_a + \ell}|\\
	\le&\ E_{e, q_a + \ell} + \frac{1}{\theta_{na}^\ell}\frac{\Vert B_dK\Pi_{K}^{\ell - 1}LQ^{i}_{q_a}\Vert\Vert C\Vert E_{e, q_a}}{N}\\
	& + \frac{1}{\theta_{na}^\ell}\frac{\Vert B_dK\Pi_{K}^{\ell - 2}LQ^{i}_{q_a + 1}\Vert\Vert C\Vert(\theta_{a} - \theta_{na}) E_{e, q_a}}{N}\\
	& + \sum_{i = 0}^{\ell - 3}\frac{1}{\theta_{na}^{i + 3}}\frac{\Vert B_dK\Pi_{K}^{i}LQ^{i}_{q_a + \ell - i - 1}\Vert\Vert CR^{-1}\Vert}{N}
	\\
	&~~ \times (\theta_0\theta_a - \theta_{na}^2) E_{e, q_a}\\
	\stackrel{\triangle}{=}& ~\tilde{E}_{e,q_a + \ell}.
	\end{align*}
	Since ${1}/{\theta_{na}}>1$,  $\{\tilde{E}_{e,q}\}$ is an increasing sequence, which contradicts the assumption that $\{|x_q - \tilde{x}_q|\}$ is a decreasing sequence.
	Therefore, it can be concluded that without ACKs, predictors at the encoder and controller sides may get asynchronized even if there is a single DoS attack.
	This causes mismatches on their quantization centers and ranges, and there exists $\hat{q} \ge q_a$ such that $E_{e, q} < |x_{q} - \tilde{x}_{q}|$ holds for all $q \ge \hat{q}$, and the state diverges eventually.
	
	We have just shown that the synchronization between decoder and encoder is essential.
	However, ACK-based protocol is not the only way to achieve this goal.
	In the absence of ACKs, this challenge can be overcome by using a deadbeat controller, and the prove will be given in the following.
	Let the number of the quantization level $N$ to be even.
	We adopt the same quantizer as in (\ref{eq:tcperror})-(\ref{eq:tcpinitial}), with $\hat{x}_q$, $e_q$, and $\hat{y}_q$ replaced by $\hat{x}_{q-1, \eta}$, $e_{q-1, \eta}$, and $\hat{y}_{q - 1, \eta}$, respectively. 
	The observer-based controller is employed only at the decoder side
	\begin{subequations}\label{controller1}
		\begin{align}
		&\hat{x}_{q, k+1} = A_d\hat{x}_{q,k} + B_du_{q,k}, & k&\le \eta - 1 \label{controller1_1}\\
		&\hat{x}_{q} = \hat{x}_{q-1, \eta} + M_q[Q(y_{q}) - \hat{y}_{q-1, \eta}], & k &= \eta \label{controller1_2}\\
		&\hat{y}_{q,k} = C \hat{x}_{q,k} \label{controller1_3}\\
		&u_{q,k} = K \hat{x}_{q,k} \label{controller1_4}
		\end{align}
	\end{subequations}
	Thanks to the ideal input channel,  
	\begin{equation*}\label{ut=uk}
	u(t) = u_{q,k}, \qquad q\Delta + k\delta \le t < q\Delta +(k+1)\delta
	\end{equation*}
	for every $q \in \mathbb{Z}_{\ge 0}$, and $k = 0, \cdots\!, \eta - 1$.
	
	Consider an arbitrary transmission interval $[q \Delta, (q + 1)\Delta)$.
	From the property (\ref{dbdb}), one gets that $\hat{x}_{q,\eta} = (A_d + B_d K)^{\eta}\hat{x}_{q} = 0$, and $\hat{y}_{q, \eta} = C\hat{x}_{q, \eta} = 0$.
	It is thus sufficient to choose the quantization center to be the origin, and predictor (\ref{controller1}) is not needed at the encoder side.
	This saves computational resources.
	
	If an attack is launched at $(q + 1)\Delta$, the decoder is not going to receive the quantized output $Q(y_{q + 1})$, and instead it will use a default zero.
	Then, it follows from (\ref{controller1_1})-(\ref{controller1_2}) that $\hat{x}_{q + 1} = \hat{x}_{q,\eta} = 0$, and $u_{q+1} = K\hat{x}_{q + 1} = 0$.
	On the other hand, in the absence of DoS attacks, since the quantization center is zero and $N$ is even, the quantized value is nonzero. 
	Therefore, the decoder receives a quantized output $Q(y_{q + 1}) \ne 0$.
	As a result, $\hat{x}_{q + 1} = \hat{x}_{q, \eta} + M(Q(y_{q+1}) - \hat{y}_{q, \eta})= MQ(y_{q + 1}) \ne 0$, and $u_{q+1} = K\hat{x}_{q + 1} \ne 0$.
	This suggests that the encoder can infer whether there is an attack or not from the input signals, thus its quantization ranges can be updated following the same scheme with the decoder.
	
	We have secured  synchronization between the encoder and decoder.
	Now, what is left behind is the system stability analysis.
	Recall that $A_d^{\eta}(I - MC)$ is schur stable, there exist constants $G_0, G_1$, and $0 <\rho <1$ such that
	\begin{equation}\label{m0m1}
	\begin{aligned}
	\big\Vert R^{\ell}\big\Vert \le G_0 \rho^{\ell},\quad
	\big\Vert {R}^{\ell}A_d^{\eta}M\big\Vert \le G_1 \rho^{\ell}.
	\end{aligned}
	\end{equation}
	Define constants 
	\begin{align*}
	\tilde{\theta}_{a} := \Vert A_d^{\eta}\Vert,~~
	\tilde{\theta}_{0} := G_0\rho + \frac{G_1\Vert C\Vert}{N},~~
	\tilde{\theta}_{na} := \rho + \frac{G_1\Vert C\Vert}{N}
	\end{align*}
	and the error bound $\{E_q:q \in \mathbb{Z}_{\ge 1}\}$ is updated by
	\begin{equation}\label{errorboundEq}
	E_{q + 1} :=
	\left\{
	\begin{aligned}
	& \tilde{\theta}_a E_q, & q\Delta \ne s_r\\
	& \tilde{\theta}_0 E_q, & (q-1)\Delta \ne s_r, q\Delta = s_r\\
	& \tilde{\theta}_{na}E_q, & (q-1)\Delta = s_r, q\Delta = s_r
	\end{aligned}
	\right..
	\end{equation}
	
	The following result is an extension of Thm. \ref{th:outputtheorem1} under an ACK-free protocol, whose proof follows from that of Thm. \ref{th:outputtheorem1}.
	\begin{theorem}\label{1convergetheorem}
		Consider system (\ref{continuoussystem}) equipped with controller in (\ref{controller1}), where $M$ and $K$ are chosen such that $R$ is schur stable and (\ref{dbdb}) is met.
		Let As. \ref{as:abca}--\ref{DoS_durationassumption} hold.
		If i) the output and input transmission periods satisfy (\ref{delta}), ii) the quantization levels $N$ is even, and obey
		\begin{equation}\label{Ncondition}
		N > \frac{G_1\Vert C\Vert}{1 - \rho}
		\end{equation}
		and, iii) DoS attacks satisfy
		\begin{equation}\label{1doscondition}
		\frac{1}{\nu_d} \le \frac{\log{(1/\tilde{\theta}_{na})}}{\log{(\tilde{\theta}_a/\tilde{\theta}_{na})}} - \frac{\log{(\tilde{\theta}_0/\tilde{\theta}_{na})}}{\log{(\tilde{\theta}_a/\tilde{\theta}_{na})}}\frac{1}{\nu_f}
		\end{equation}
		then the system is exponentially stable under the encoding scheme with error bound $\{E_q:q \in \mathbb{Z}_{\ge 1}\}$ constructed by (\ref{errorboundEq}).
	\end{theorem}
	\begin{figure}[b]
		\centering
		\includegraphics[width=9cm]{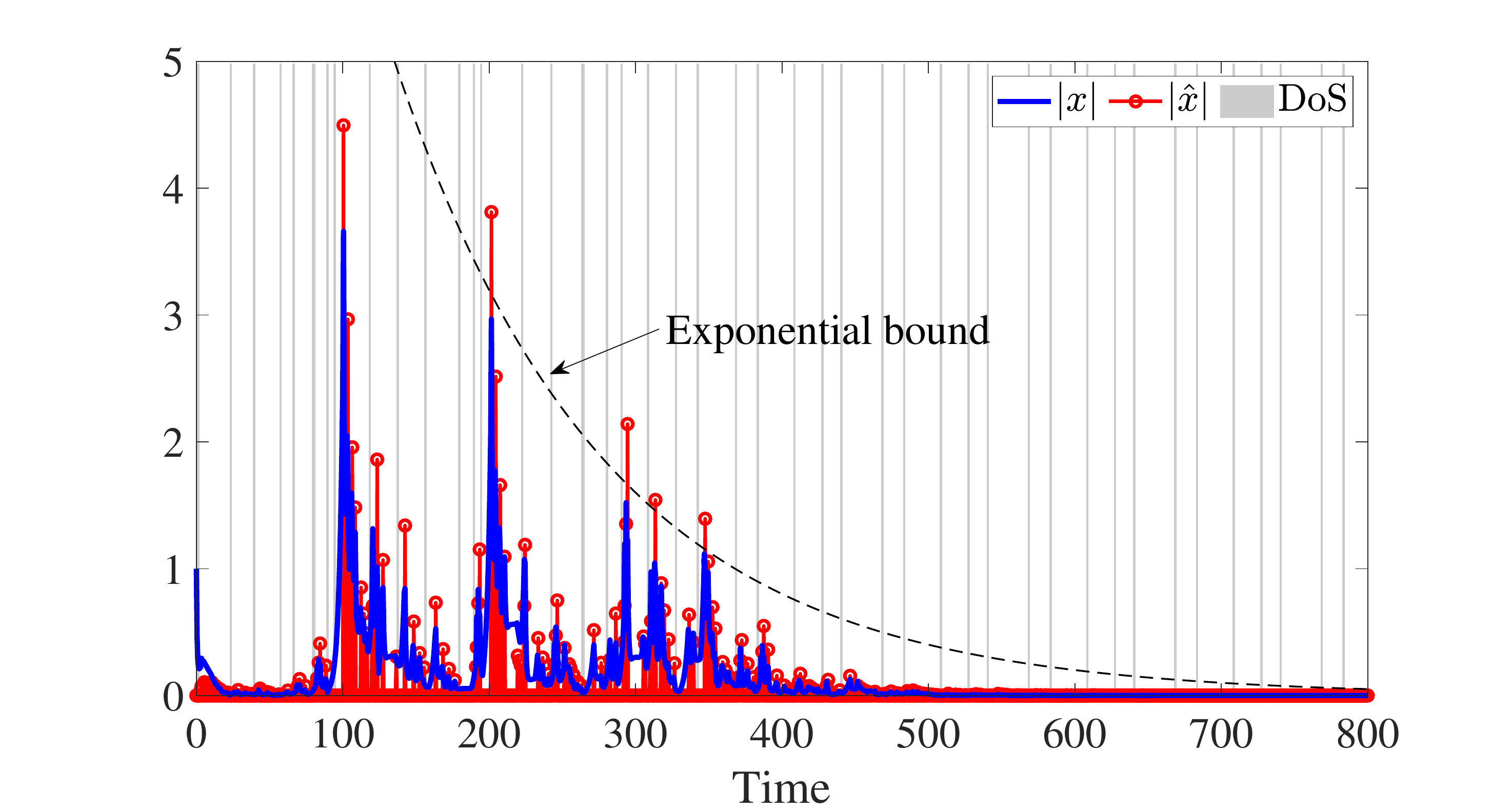}\\
		\caption{Maximum norm of state $x$ and its estimate $\hat{x}$ with controller (\ref{abdoscontroller}).}\label{figdoubledblqnormx}
		\centering
	\end{figure}
	\begin{figure}
		\centering
		\includegraphics[width=9cm]{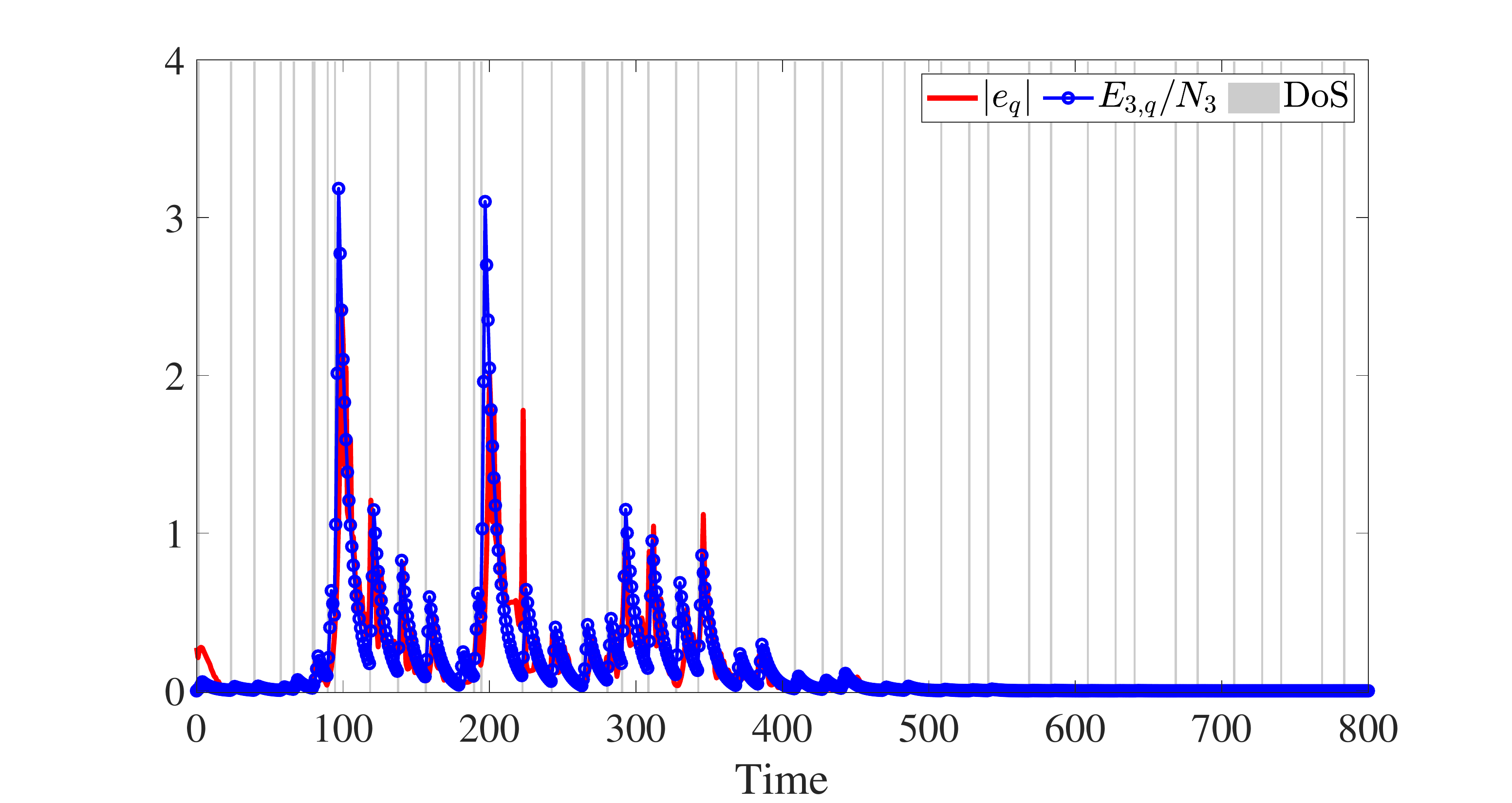}\\
		\caption{Relationship between normalized quantization range $E_{3,k}/N_3$ and actual error $|y_q - Q_1(\hat{y}_q)|$ with controller (\ref{abdoscontroller}).}\label{figdoubledblqnormE}
		\centering
	\end{figure}
	\begin{figure}
		\centering
		\includegraphics[width=9cm]{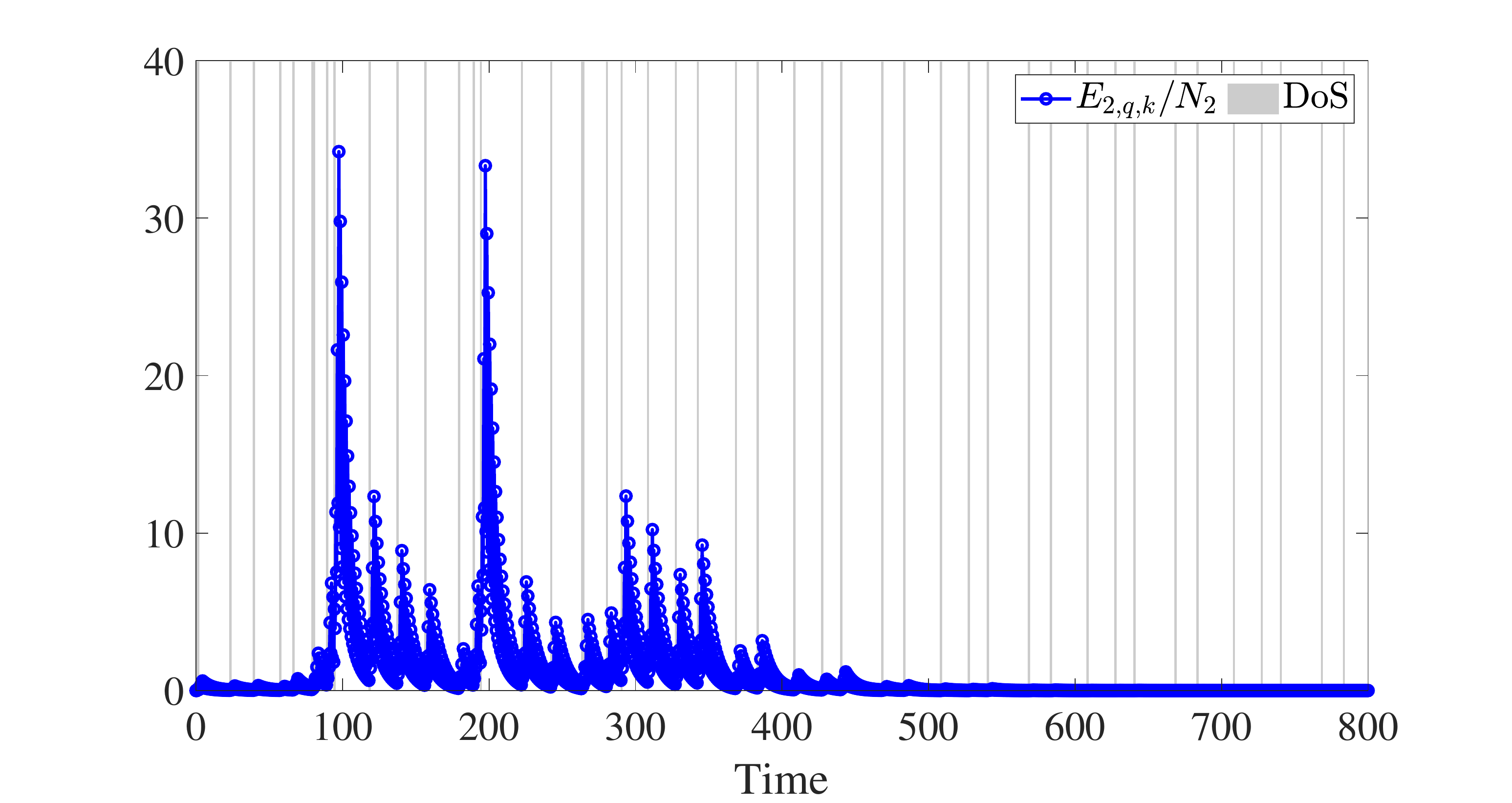}\\
		\caption{Normalized quantization range $E_{2,q,k}/N_2$ with controller (\ref{abdoscontroller}).}\label{figdoubledblqnormE2}
		\centering
	\end{figure}
	\begin{figure}
		\centering
		\includegraphics[width=9cm]{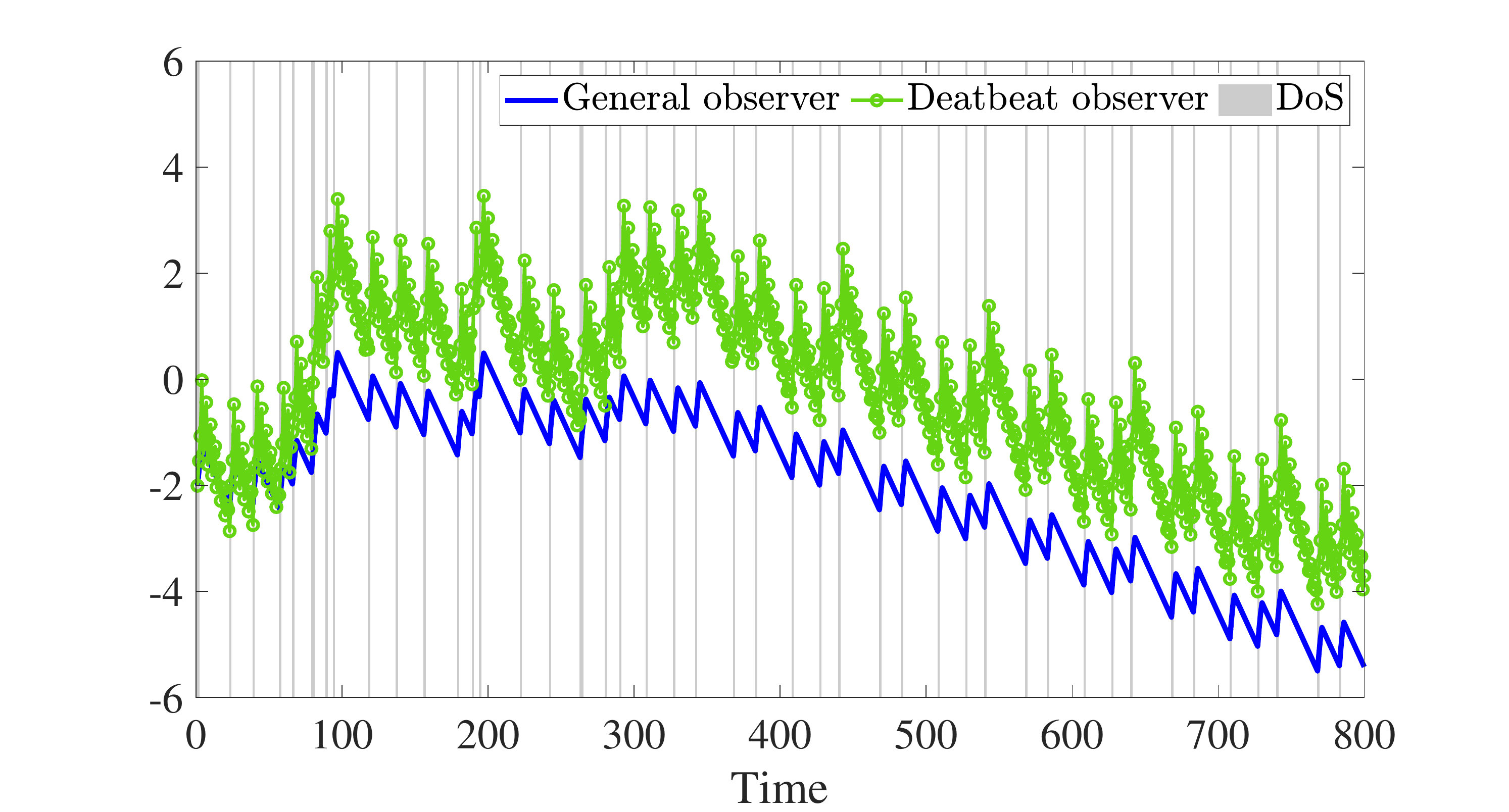}\\
		\caption{Normalized quantization ranges $E_{3,q}/N_3$ from using general observer gain and deadbeat observer gain in log space.}\label{figlogdblq}
		\centering
	\end{figure}
	\begin{figure}
		\centering
		\includegraphics[width=9cm]{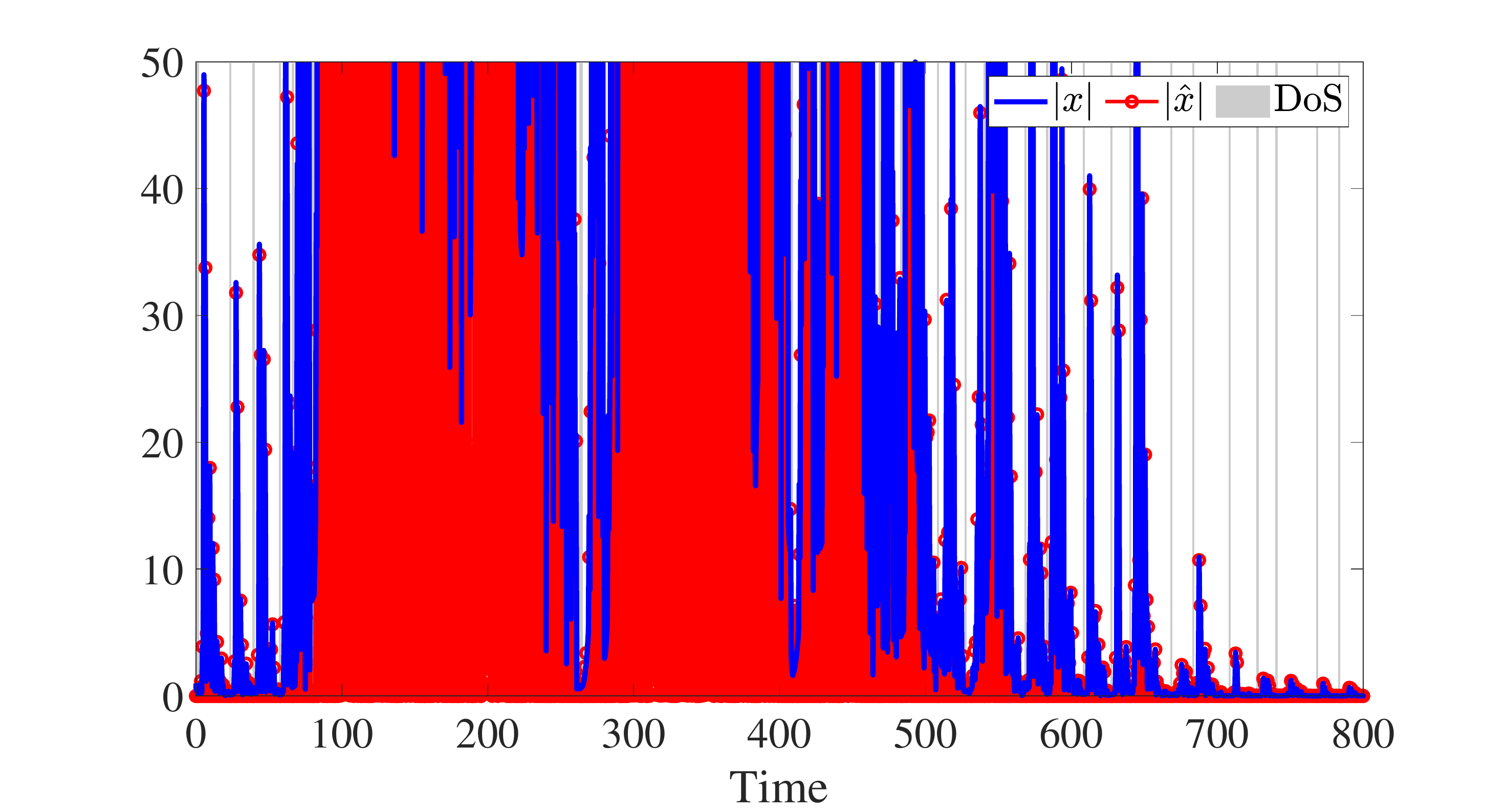}\\
		\caption{Maximum norm of state $x$ and its estimate $\hat{x}$ with controller (\ref{abdoscontroller}) and deadbeat observer gain.}\label{figdoubledbdbnormx}
		\centering
	\end{figure}
	\section{Numerical Example}
	\label{simulation}
	A linearized model of the unstable batch reactor in \cite{8880482} is given by $\dot{x}(t) = A x(t) + B u(t)$ and $y = C x(t)$, where
	\begin{align*}
	A :=
	\left[
	\begin{matrix}
	1.38 & -0.2077 & 6.715 & -5.676\\
	-0.5814 & -4.29 & 0 & 0.675\\
	1.067 & 4.273 & -6.654 & 5.893\\
	0.048 & 4.273 & -1.343 & -2.104
	\end{matrix}
	\right],\\
	B := \left[
	\begin{matrix}
	0 & 0\\
	5.679 & 0\\
	1.136 & -3.146\\
	1.136 & 0
	\end{matrix}
	\right],
	C :=
	\left[
	\begin{matrix}
	1 & 0 & 1 & -1\\
	0 & 1 & 0 & 0
	\end{matrix}
	\right].
	\end{align*}
	This system $(A, B, C)$ is observable and controllable with $\eta =\mu =  2$.
	Let the output transmission period $\Delta = 0.2$, so $\delta = \Delta/\eta = 0.1$.
	Choosing matrix $K$, such that (\ref{dbdb}) is met, i.e.,
	\begin{align*}
	& K :=
	\left[
	\begin{matrix}
	1.0106 & -1.5661 & 0.0385 & -4.0366\\
	8.1074 & -0.0347 & 4.3337 &- 3.6241
	\end{matrix}
	\right].
	\end{align*}
	Calculating the gain of the steady-state Kalman filter  
	\begin{align*}
	M :=
	\left[
	\begin{matrix}
	0.5534 & -0.0249\\
	-0.0287 & 0.0396\\
	0.1489 & 0.0892\\
	0.0810 & 0.0931
	\end{matrix}
	\right].
	\end{align*}
	
	We first present the time responses when both input and output channels suffer from the network phenomena.
	Applying Thm. \ref{2convergetheorem}, when both the quantization levels $N_2$ and $N_3$ go to infinity, the duration bound $1/\nu_d$ and the frequency bound $1/\nu_f$ of DoS attacks approach to the line $\frac{1}{\nu_d} \approx -0.5544\frac{1}{\nu_f} + 0.2707$.
%	Based on (\ref{2Ncondition}), the quantization levels $N_1 \ge 1, N_2 > 188.1898$, and $N_3 > 14.6488$, hence  we set $(N_1, N_2, N_3) = (1, 551, 451)$.
	According to (\ref{2doscondition}), if
	$\frac{1}{\nu_d} < -2.0380\frac{1}{\nu_f} + 0.2269$,
	then the closed-loop system with encoding schemes (\ref{2E1q})-(\ref{2E3q}) is stabilized.
	Over a simulation horizon of $160$s ($800$ time-step), DoS attacks (the gray shades) are generated randomly with $\Phi_d = 47$ and $\Phi_f = 44$.
	Setting $\kappa_d = 3, \nu_d = 18, \kappa_f = 2, \nu_f = 19$, condition (\ref{2doscondition}) holds, i.e., $1/\nu_d = 0.056 < 0.119$.
	Figs. \ref{figdoubledblqnormx} and \ref{figdoubledblqnormE} illustrate the time response in this situation.
	Since the condition in Thm. \ref{2convergetheorem} is satisfied, the maximum norm of the state converges, and the bound $E_{3,q}$ exponentially decreases.
	Fig. \ref{figdoubledblqnormE} depicts that $E_{3,q}$ shares the same trend with $|y_q - Q_1(\hat{y}_{q})|$, and Fig. \ref{figdoubledblqnormE2} demonstrates the evolution of the quantization step size $E_{2,q,k}/N_2$, which jumps up and down within an output transmission period, and decreases in general.
	Difference between the trend of $E_{3,q}/N_3$ and $E_{2,q,k}/N_2$ lies in the property of $\Vert R \Vert$ and $\Vert \bar{R}\Vert$.
	Fig. \ref{figlogdblq} compares the quantization step size $E_{3,q}/N_3$ of a general observer gain (blue line), such that $R$ is schur stable, and the deadbeat observer gain (dot marked green line), namely $R^{\mu} = 0$.
	This panel illustrates that although $E_{3,q}$ responds faster under deadbeat observer, the large quantization step size results in large overshoot of the state; see Fig. \ref{figdoubledbdbnormx}, which confirms Rmk. \ref{2dbdbremark}.
	\begin{figure}
		\centering
		\includegraphics[width=9cm]{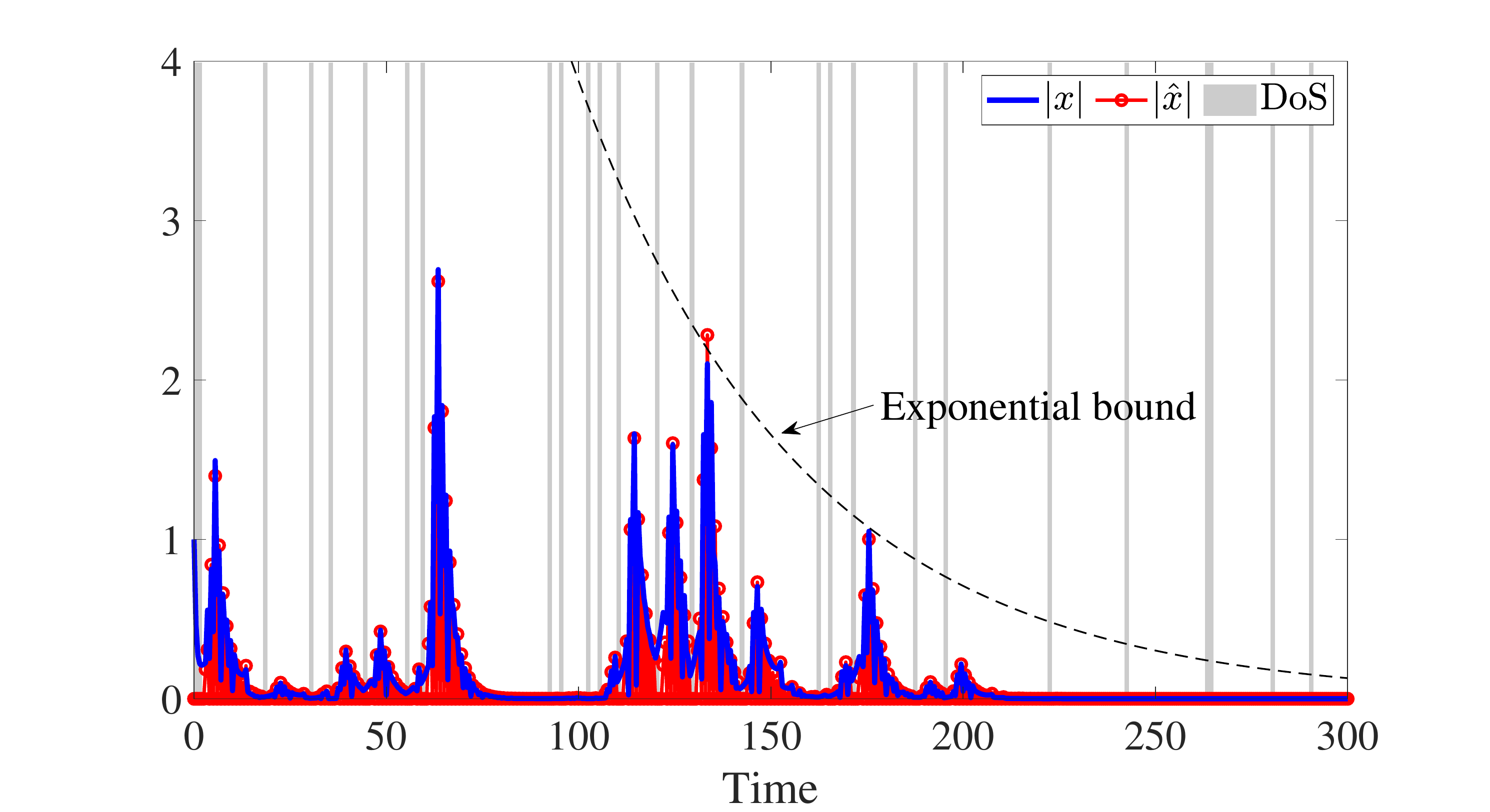}\\
		\caption{Maximum norm of state $x$ and its estimate $\hat{x}$ with controller (\ref{controller1}).}\label{figsiglenormx}
		\centering
	\end{figure}
	\begin{figure}
		\centering
		\includegraphics[width=9cm]{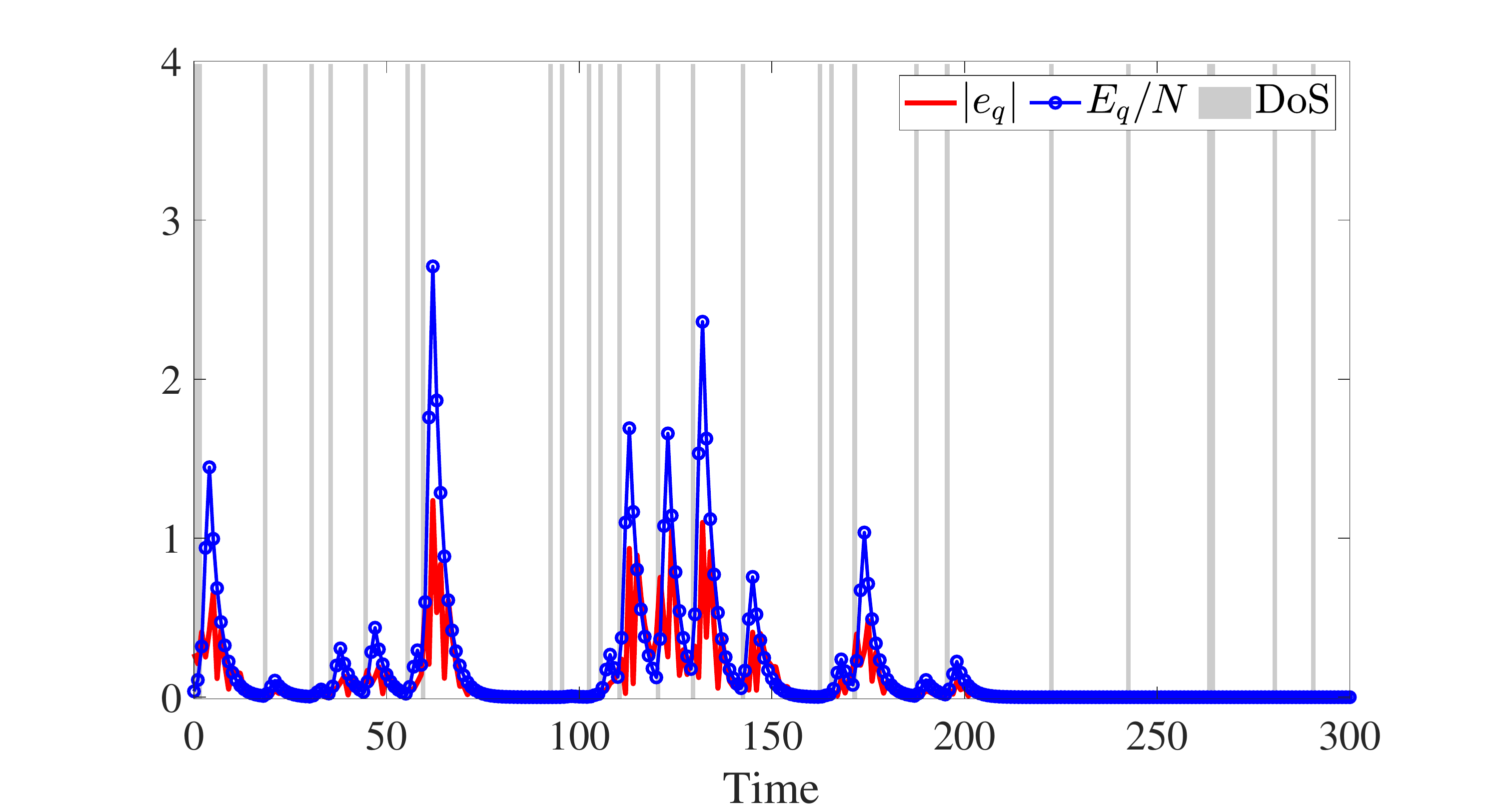}\\
		\caption{Relationship between normalized quantization range $E_q/N$ and actual error $|e_q|$ with controller (\ref{controller1}).}\label{figsigleerrorE}
		\centering
	\end{figure}
	\begin{figure}
		\centering
		\includegraphics[width=9cm]{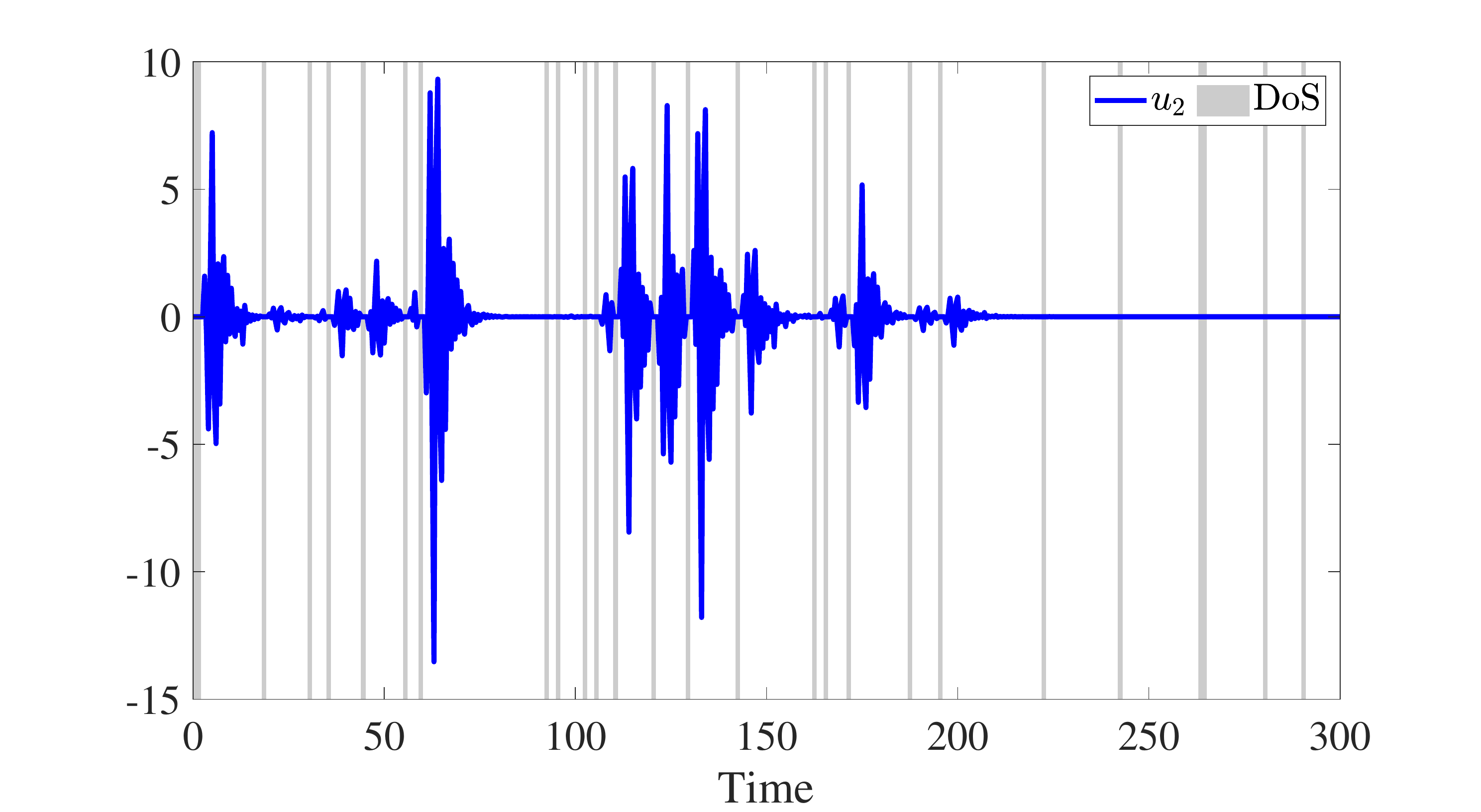}\\
		\caption{Input signal $u_{q,k}$ with controller (\ref{controller1}).}\label{figsigleu}
		\centering
	\end{figure}
	
	Next, consider network phenomena only at output channel.
	From (\ref{Ncondition}), the quantization levels satisfies $N > 6.957$, also since $N$ is even, we set $N = 100$.
	Over a simulation horizon of $60$s ($300$ time-step), generating DoS attacks randomly with $\Phi_d = 27$ and $\Phi_f = 25$.
	Setting $\kappa_d = 1, \nu_d = 11, \kappa_f = 1, \nu_f = 11$, so condition (\ref{1doscondition}) is met with $1/\nu_d = 0.01 < 0.198$, and convergence of the state is presented in Figs. \ref{figsiglenormx} and \ref{figsigleerrorE}.
	Further, Fig. \ref{figsigleu} shows that when a DoS attack happens, the control input is set to zero immediately, which verifies the effectiveness of our method.
	
	\section{Conclusions}\label{conclusion}
%	This paper considered the problem of stabilizing linear systems under DoS attacks and with limited bandwidth.
%%	 first at both input and output channels, then at only the output channel.
%	In this setting, designing resilient encoding schemes while accounting for DoS attacks is proven challenging in two aspects.
%	The first is with the coupling between the encoding strategies, and the other is with the asynchronization between encoder and decoder.
%%	In both cases, designing resilient encoding schemes while accounting for DoS attacks is proven challenging. 
%	A novel structure, consisting of a deadbeat controller and a designed transmission protocol, was put forward to mitigate such difficulties.
%	When both input and output channels have network phenomena, this structure decoupled the encoding schemes among the input, the output, and the estimated output.
%	This property has been exemplified by deriving encoding schemes and stability conditions for controllers applying such structure.
%	In addition,
%	if only the output channel is corrupted by DoS attacks, the proposed structure has been proved to guarantee the synchronization between encoder and decoder under an acknowledgment-free protocol.
%	Finally, a numerical example has been presented to verify the effectiveness of our approach.
%	Future developments will focus on generalizing the results to more general controllers under acknowledgment-free protocol.
	This paper considered the problem of stabilizing networked control systems in the presence of DoS attacks and limited data rates. To overcome the network-induced challenges, a structure consisting of a deadbeat controller and a transmission protocol which are carefully co-designed based on the system controllability index, was proposed to address the network-induced challenges. 
	Specifically, when both input and output channels are subject to the network phenomena, it was shown that the proposed structure can decouple and thus allow for separate design of encoding schemes for the input, output, and estimated output signals. Furthermore, easy-to-check conditions were derived such that exponential stability of the closed-loop system under this structure is ensured. On the other hand, when only the output channel is subject to the network phenomena, the proposed structure was shown able to guarantee synchronization between the encoder and decoder under an ACK-free protocol.
	Finally, a numerical example was presented to verify the effectiveness of our approach as well as the correctness of our theory.
	Future developments will focus on generalizing the results to more general systems and controllers under ACK-free protocols.
	
	%%%%%%%%%%%%%%%%%%%%%%%%%%%%%%%%%%%%%%%%%%%%%%%%%%%%%%%%%%%%%%%%%%%%%%%%%%%%%%%%
	\bibliographystyle{IEEEtran}
	
	\bibliography{bible1}
	
	\begin{IEEEbiography}[{\includegraphics[width=1in,height=1.25in,clip,keepaspectratio]{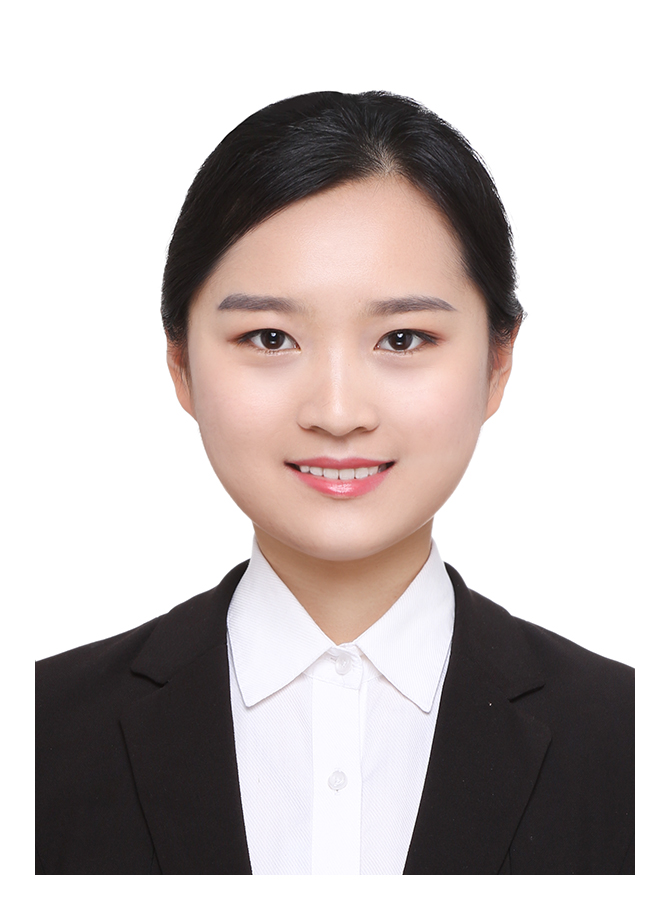}}]{Wenjie Liu} received the bachelor's degree in Automation from Tianjin University, Tianjin, China, in 2019.
		She is currently pursuing the Ph.D. degree in control science and engineering with the School of Automation, Beijing Institute of Technology, Beijing, China.		
		Her current research interests include cyber–physical systems and network control under communication constraints.
	\end{IEEEbiography}
	\begin{IEEEbiography}[{\includegraphics[width=1in,height=1.25in,clip,keepaspectratio]{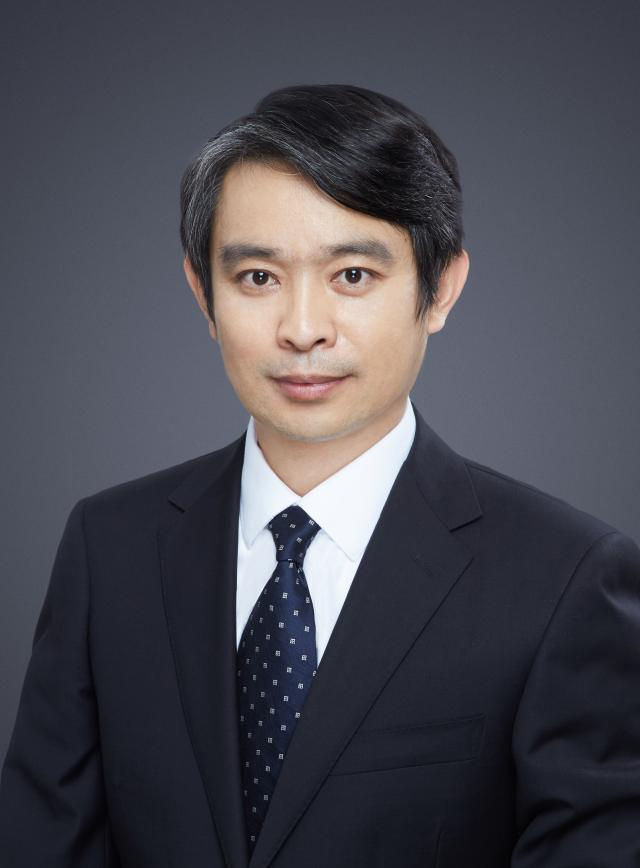}}]{Jian Sun} received the bachelor's degree from the Department of Automation and Electric Engineering, Jilin Institute of Technology, Changchun, China, in 2001, the master's degree from the Changchun Institute of Optics, Fine Mechanics and Physics, Chinese Academy of Sciences (CAS), Changchun, China, in 2004, and the Ph.D. degree from the Institute of Automation, CAS, Beijing, China, in 2007.
		
		He was a Research Fellow with the Faculty of Advanced Technology, University of Glamorgan, Pontypridd, U.K., from 2008 to 2009. He was a Post-Doctoral Research Fellow with the Beijing Institute of Technology, Beijing, from 2007 to 2010. In 2010, he joined the School of Automation, Beijing Institute of Technology, where he has been a Professor since 2013. His current research interests include networked control systems, time-delay systems, and security of cyber-physical systems.
		
		Dr. Sun is an Editorial Board Member of the IEEE Transactions on Systems, Man and Cybernetics: Systems, the Journal of Systems Science \& Complexity, and Acta. Automatica Sinica.
	\end{IEEEbiography}	
	\begin{IEEEbiography}[{\includegraphics[width=1in,height=1.5in,clip,keepaspectratio]{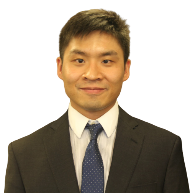}}] {Gang Wang} (M'18) received a B.Eng. degree in Automatic Control in 2011, and a Ph.D. degree in Control Science and Engineering in 2018, both from the Beijing Institute of Technology, Beijing, China. He also received a Ph.D. degree in Electrical and Computer Engineering from the University of Minnesota, Minneapolis, USA, in 2018, where he stayed as a postdoctoral researcher until July 2020. Since August 2020, he has been a professor with the School of Automation at the Beijing Institute of Technology. 
		
		His research interests focus on the areas of signal processing, control, and reinforcement learning with applications to cyber-physical systems and multi-agent systems. 
		He was the recipient of the Excellent Doctoral Dissertation Award from the Chinese Association of Automation in 2019, the Best Student Paper Award from the 2017 European Signal Processing Conference, and the Best Conference Paper at the 2019 IEEE Power \& Energy Society General Meeting. He is currently on the editorial board of Signal Processing.
		
	\end{IEEEbiography}

		\begin{IEEEbiography}[{\includegraphics[width=1in,height=1.25in,clip,keepaspectratio]{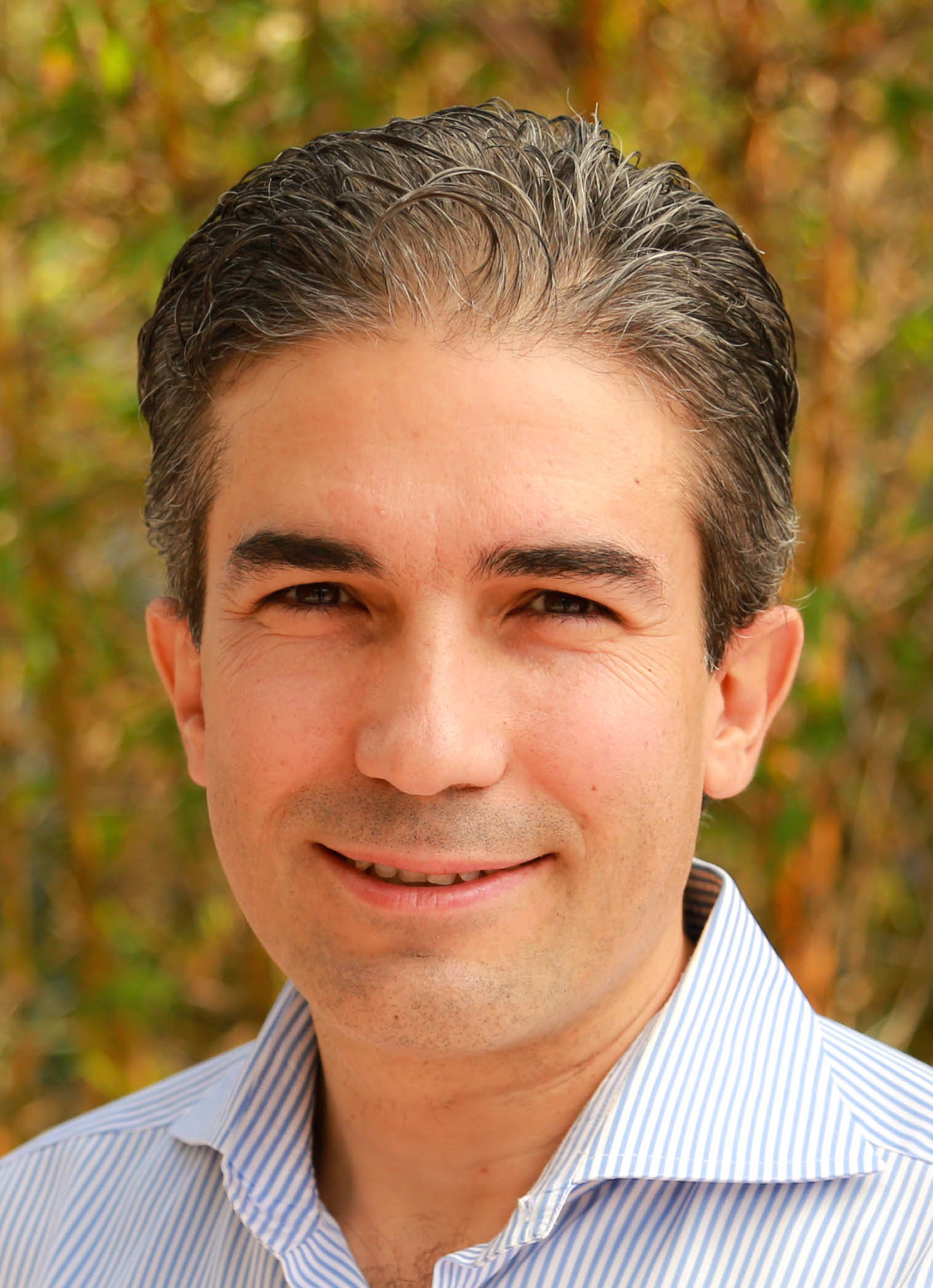}}]{Francesco Bullo} (IEEE S'95-M'99-SM'03-F'10) is a Professor with the Mechanical Engineering Department and the Center for Control, Dynamical Systems and Computation at the University of California, Santa Barbara. He was previously associated with the University of Padova, the California Institute of Technology, and the University of Illinois.
			 
			His research focuses on modeling, dynamics and control of multi-agent network systems, with applications to robotic coordination, power systems, distributed computing and social networks.
			Previous work includes contributions to geometric control, Lagrangian systems, vehicle routing, and motion planning. 
			He has published more than 300 papers in international journals, books, and refereed conferences. 
			He is the coauthor, with Andrew D. Lewis, of the book ``Geometric Control of Mechanical Systems'' (Springer, 2004, 0-387-22195-6), with Jorge Cortés and Sonia Martínez, of the book “Distributed Control of Robotic Networks” (Princeton, 2009, 978-0-691-14195-4), with Stephen L. Smith of the book “Lectures on Robotics Planning and Kinematics” (SIAM, 2019, under contract); and of the book “Lectures on Network Systems” (Kindle Direct Publishing, 2020, v1.4, 978-1986425643).
			He received best paper awards for his work in IEEE Control Systems, Automatica, SIAM Journal on Control and Optimization, IEEE Transactions on Circuits and Systems, and IEEE Transactions on Control of Network Systems. 
			
			He is a Fellow of IEEE, IFAC, and SIAM. He has served on the editorial boards of IEEE, SIAM, and ESAIM journals, and serves as 2018 IEEE CSS President.
			He is serving as Chair of the SIAM Activity Group on Control and Systems Theory, for 2020-2021.
%			Prof. Bullo has served the IEEE Control Systems Society in various roles: as President Elect / President / President Past during the triennium 2017–2019, as 2011-2012 Vice-President for Technical Activities, as 2013-2014 Vice-President for Publications, as 2007-2009 Elected Member of the Board of Governors, and as Program Chair for the 2016 IEEE Conference in Decision and Control. Additionally, he served on the Editorial Boards of IEEE Transactions on Automatic Control, ESAIM: Control, Optimization, and the Calculus of Variations, SIAM Journal of Control and Optimization, and Mathematics of Control, Signals, and Systems. He is serving as Chair of the SIAM Activity Group on Control and Systems Theory, for 2020-2021.
	\end{IEEEbiography}
	
	\begin{IEEEbiography}[{\includegraphics[width=1in,height=1.25in,clip,keepaspectratio]{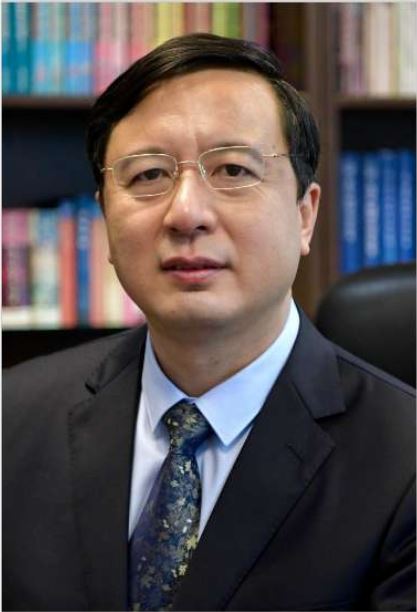}}]{Jie Chen} (F'19) received his B.Sc., M.Sc., and the Ph.D. degrees in control theory and control engineering from the Beijing Institute of Technology, Beijing, China, in 1986, 1996, and 2001, respectively. From 1989 to 1990, he was a visiting scholar at the California State University, Long Beach, California, USA. From 1996 to 1997, he was a research fellow in the School of Engineering at the University of Birmingham, Birmingham, UK. 
		
		He is a Professor with the School of Automation, Beijing Institute of Technology, where he serves as the Director of the Key Laboratory of Intelligent Control and Decision of Complex Systems. He also serves as the President of Tongji University, Shanghai, China. 
		His research interests include complex systems, multiagent systems, multiobjective optimization and decision, and constrained nonlinear control.
		
		Prof. Chen is currently the Editor-in-Chief of Unmanned Systems and
		the Journal of Systems Science and Complexity. He has served on the editorial boards
		of several journals, including the
		IEEE Transactions on Cybernetics, International Journal of Robust
		and Nonlinear Control, and Science China Information Sciences. He is a Fellow of IEEE, IFAC, and a member of the Chinese Academy of Engineering.
		 
%		Prof. Chen served as a managing editor for the Journal of Systems Science \& Complexity, and an associate/subject editor for several other international journals, including the IEEE Transactions on Cybernetics, and the International Journal of Robust and Nonlinear Control. He was a recipient of several prestigious awards, including one second-grade  prize of the National Natural Science Award of China (2014), and two second-grade prizes of the National Science and Technology Progress Award of China (2009,	2011). He is an IEEE Fellow, IFAC Fellow, and a member of the Chinese Academy of Engineering.
		
%		
%		His main research interests include intelligent control and decision in complex systems, multi-agent systems, nonlinear control, and optimization. He has (co-)authored 4 books and more than 200 research papers. He served as a managing editor for the Journal of Systems Science \& Complexity, and an associate/subject editor for several other international journals, including the IEEE Transactions on Cybernetics, and the International Journal of Robust and Nonlinear Control. He was a recipient of several prestigious awards, including one second-grade prize of the National Natural Science Award of China (2014), and two second-grade prizes of the National Science and Technology Progress Award of China (2009,	2011). He is an IEEE Fellow, IFAC Fellow, and a member of the Chinese Academy of Engineering.
		\end{IEEEbiography}
\end{document}